\documentclass{acm_proc_article-sp}

\usepackage{amsmath}
\usepackage{amssymb}
\usepackage{color}
\usepackage{listings}
\usepackage{graphicx} 

\let\emptyset=\varnothing
\DeclareMathOperator{\sep}{\mid\!\mid}

\newcommand{\ie}{\textit{i.e., }}
\newcommand{\eg}{\textit{e.g., }}

\newcommand{\inputv}[1]{\check{{#1}}}
\newcommand{\outputv}[1]{\hat{{#1}}}

\newcommand*{\len}{\mathit{len}}
\newcommand*{\nat}{\mathbb{N}}
\newcommand*{\integers}{\mathbb{Z}}
\newcommand{\correctins}{\mathcal{C}}
\newcommand*{\dom}{\mathit{\mathit{dom}}}
\newcommand*{\codom}{\mathit{\mathit{rng}}}
\newcommand*{\pl}{\mathit{\mathit{pl}}}
\newcommand*{\la}{\mathop{\leftarrow}\limits}
\newcommand*{\ra}{\mathop{\rightarrow}\limits}
\newcommand*{\lra}{\mathop{\longrightarrow}\limits}

\newcommand*{\block}[1]{\fbox{${#1}$}\hspace*{1ex}}
\newcommand*{\tto}{\rightrightarrows}
\newcommand*{\ttob}{{\tto\begin{smallmatrix}b_1\\\cdots\\b_n\end{smallmatrix}}}
\newcommand*{\vars}{\mathit{vars}}
\newcommand{\clppl}{\mathrm{CLP}(\mathbb{PL})}

\newtheorem{proposition}{Proposition}
\newtheorem{theorem}{Theorem}
\newtheorem{claim}{Claim}
\newdef{definition}{Definition}
\newdef{example}{Example}

\begin{document}

%%%%%%%%%%%%%%%%%%%%%%%%%%%%%%%%%%%%%%%%%%%%%%%%%%%%%%%%%%%%%%%%%

\title{Non-Termination Analysis of Java Bytecode}

\numberofauthors{3}
\author{
% 1st author
  \alignauthor \'Etienne Payet\\
  \affaddr{LIM-IREMIA}\\
  \affaddr{Universit\'e de la R\'eunion}\\
  \affaddr{France}\\
  \email{etienne.payet@univ-reunion.fr}
% 2nd author
  \alignauthor Fred Mesnard\\
  \affaddr{LIM-IREMIA}\\
  \affaddr{Universit\'e de la R\'eunion}\\
  \affaddr{France}\\
  \email{frederic.mesnard@univ-reunion.fr}
% 3rd author
  \alignauthor Fausto Spoto\\
  \affaddr{Dipartimento di Informatica}\\
  \affaddr{Universit\`a di Verona}\\
  \affaddr{Italy}\\
  \email{fausto.spoto@univr.it}
}

\maketitle
  
\begin{abstract}
  We introduce a fully automated static
  analysis that takes a sequential Java bytecode program
  $P$ as input and attempts to prove that there exists an
  infinite execution of $P$. The technique consists in 
  compiling $P$ into a constraint logic program $P_{\mathit{CLP}}$
  and in proving non-termination of $P_{\mathit{CLP}}$;
  when $P$ consists of instructions that are \emph{exactly}
  compiled into constraints, the non-termination
  of $P_{\mathit{CLP}}$ entails that of $P$.
  Our approach can handle method calls; to the best of our
  knowledge, it is the first static approach for Java bytecode
  able to prove the existence of infinite recursions.
  We have implemented our technique inside the Julia analyser.
  We have compared the results of Julia on a set of
  113 programs with those provided by AProVE and Invel, the only
  freely usable non-termination analysers comparable to ours
  that we are aware of. Only Julia could detect non-termination
  due to infinite recursion.
\end{abstract}  




%%%%%%%%%%%%%%%%%%%%%%%%%%%%%%%%%%%%%%%%%%%%%%%%%%%%%%%%%%%%%%%%%

%=======================
\section{Introduction}
%=======================
%
In this paper, we address the issue of automatically proving
non-termination of sequential Java bytecode programs.
We describe and implement a static analysis that takes a program $P$
as input and attempts to prove that there exists an infinite execution
of $P$.
It is well-known that termination of computer programs is an undecidable
property, hence a non-termination analyser for Java bytecode can be used
to complement any existing termination analyser,
\eg{} AProVE~\cite{Otto10}, COSTA~\cite{AlbertAGPZ12} or
Julia~\cite{SpotoMP10}.
Research in non-termination has mainly been focused
on logic programs
\cite{Bol91b,DeSchreye90a,Payet09a,Payet06a,Shen03,Shen01a}
and term rewriting systems
\cite{Giesl05,Payet08a,Waldmann04,Waldmann07,Zankl07b,Zantema05}.
Only a few recent papers address the problem of proving
non-termination of imperative programs:
% \cite{BurnimJSS09,CarbinMKR11} introduce
% techniques for dynamically analysing a running program to prove that
% it is non-terminating,
\cite{BrockschmidtSOG} considers Java bytecode,
\cite{Gupta08} considers programs written in the C language
and \cite{VelroyenR08} considers imperative programs
that can be described as logical formul\ae{} written in a simple
while-language.

\subsection{Contributions}
%==========================
In~\cite{Payet09b}, we presented a first experimentation
with the automatic derivation of non-termination proofs for
Java bytecode programs. There, we started from the results
introduced in a preliminary version of~\cite{SpotoMP10}
where the original Java bytecode program $P$ is translated
into a constraint logic program $P_{\mathit{CLP}}$ whose
termination entails that of $P$. We had the idea of carrying
out a very simple non-termination analysis of $P_{\mathit{CLP}}$
using earlier results introduced in~\cite{Payet09a}.
During our experiments with non-terminating Java bytecode
programs, we made the empirical observation that the non-termination
of $P_{\mathit{CLP}}$ entails that of $P$ when $P_{\mathit{CLP}}$
is an \emph{exact} translation of $P$. We only introduced a
very intuitive and non-formal definition of \emph{exactness}
and we did not give any formal proof of this entailment.
In this paper, we provide the formal definitions and results
that are missing in~\cite{Payet09b}; the corresponding formal
proofs are available in the long version at~\cite{PayetMP12long}.
We also provide a non-termination criterion that works for
method calls and recursion, together with a new experimental
evaluation of our results over a set of 113 Java bytecode programs.

The technique we apply for proving non-termination of
$P_{\mathit{CLP}}$ is an improvement of a simple sufficient 
condition for linear binary CLP programs~\cite{Payet09a}.
This improved condition (Proposition~\ref{proposition:unit-loop})
is another contribution of this paper.
Our main result (Theorem~\ref{theorem:non-term})
is independent of the non-termination detection procedure.
Let us point out that there is no perfect non-termination criterion
for the CLP programs we consider:
\cite{Bradley05polyrankingEV} shows that the
termination of binary CLP programs with
linear constraints over the integers
is undecidable. However some interesting
subclasses have been recently investigated.
For instance, when all the constraints are of the form
$x > y$ or $x \geq y$, termination
of the binary CLP program is decidable
\cite{BenAmramMCsZ}. So when the generated
CLP program falls into this class,
we could replace our general non-termination test by a
decision procedure for non-termination.

Our results are fully implemented inside the Julia static analyser,
that we used for conducting the experiments.
Julia is a commercial product (\texttt{http://www.juliasoft.com}).
Its non-termination analysis can be freely used through the
web interface~\cite{JuliaWeb}, whose power is limited by a
time-out and a maximal size of analysis.

\subsection{Related Works}
%===========================
To the best of our knowledge, only
\cite{BrockschmidtSOG,Gupta08,VelroyenR08} introduce methods
and implementations that are directly comparable to the results
of this paper. 

In~\cite{BrockschmidtSOG}, the program $P$ under analysis is
first transformed into a \emph{termination graph} that finitely
represents all runs through the program. Then, a
term rewrite system is generated from the graph and existing
techniques from term rewriting are used to prove non-termination
of the rewrite system. This approach has been successfully
implemented inside the AProVE analyser~\cite{AProVEWeb,GieslST}.
Note that the rewrite system generated from the termination graph
is an \emph{abstraction}~\cite{CousotC77} of $P$; the technique
that we present in this paper also computes an abstraction of $P$
but a difference is that our abstraction consists of a constraint
logic program $P_{\mathit{CLP}}$ instead of a term rewrite system.

The technique described in~\cite{Gupta08} is a combination
of dynamic and static analysis. It consists in generating
\emph{lassos} that are checked for \emph{feasibility}.
A lasso consists of a finite program path called stem,
followed by a finite program path called loop; it is
feasible when an execution of the stem can be followed
by infinitely many executions of the loop. Lassos are
generated through a dynamic execution of the program on
concrete as well as symbolic inputs;
symbolic constraints are gathered during this execution
and are used for expressing the feasibility of the
lassos as a constraint satisfaction problem.
This technique has been implemented inside the TNT
non-termination analyser for C programs. The analysis that
we present in this paper also looks for feasible lassos: 
it tries to detect some loops that have an infinite execution
from some input values and to prove that these values are
reachable from the main entry point of the program
(Proposition~\ref{proposition:compound-loops}).
A difference is that our technique does not combine static
with dynamic analysis. Another difference is that the approach
of~\cite{Gupta08} provides a bit-level analysis which is able
to detect non-termination due, \eg  to arithmetic overflow.

In~\cite{VelroyenR08}, the authors consider
a simple while-language that is used to describe programs
as logical formul\ae. The non-termination of the program $P$
under analysis is expressed as a logical formula involving
the description of $P$. The method then consists in proving
that the non-termination formula is true by constructing
a proof tree using a Gentzen-style sequent calculus.
The rule of the sequent calculus corresponding to the \texttt{while}
instruction uses invariants, that have to be generated by an
external method.
Hence, \cite{VelroyenR08} introduces several
techniques for creating and for scoring the invariants according to
their probable usefulness; useless invariants are discarded
(invariant filtering). The generated invariants are stored inside
a queue ordered by the scores.
The algorithms described in~\cite{VelroyenR08} have
been implemented inside the Invel non-termination analyser for Java
programs~\cite{InvelWeb}. Invel uses the KeY~\cite{KeYBook2007}
theorem prover for constructing proof trees. As far as we know,
it was the first tool for automatically proving non-termination
of imperative programs.

One of the main differences between the techniques introduced
in~\cite{Gupta08,VelroyenR08} and ours is that we first construct
an abstraction of the program under analysis and then we keep
on reasoning on this abstraction only.
The algorithms presented in~\cite{Gupta08,VelroyenR08}
model the semantics of the concrete program more accurately.
They hence do not suffer from some lack of precision that we face; 
we were not able to exactly translate some bytecode instructions
into constraints, therefore our method fails on the programs that
include these instructions. On the other hand, the techniques 
that directly consider the original concrete program are generally
time consuming and they do not scale very well. 
% The experiments that we present in this paper illustrate
% these considerations very clearly. 

Finally, a major difference between our approach and that
of~\cite{BrockschmidtSOG,VelroyenR08} is that
we are able to detect non-termination due to infinite
recursion, whereas \cite{BrockschmidtSOG,VelroyenR08}
are not. Our experiments illustrate this consideration
very clearly. Note that the approach in~\cite{Gupta08} can deal
with non-terminating recursion.

\subsection{Organisation of this Paper}
%=======================================
The rest of this paper is organised as follows.
Section~\ref{section:preliminaries} introduces the basic formal
material borrowed from~\cite{SpotoMP10}.
Section~\ref{section:exact-approximations} provides
a formal definition of exactness for the abstraction of
a Java bytecode instruction into a linear constraint.
In Section~\ref{section:compilation}, we show how to
automatically generate a constraint logic program
$P_{\mathit{CLP}}$ from a Java bytecode program $P$
so that the non-termination of $P_{\mathit{CLP}}$ entails
that of $P$.
Section~\ref{section:clp-nontermination} deals with proving
non-termination of $P_{\mathit{CLP}}$; it provides an improvement
of a non-termination criterion that we proposed in~\cite{Payet09a}.
Section~\ref{section:experiments} describes our experiments
on a set of 113 non-terminating programs obtained from different
sources.
Section~\ref{section:conclusion} concludes the paper.

%%%%%%%%%%%%%%%%%%%%%%%%%%%%%%%%%%%%%%%%%%%%%%%%%%%%%%%%%%%%%%%%%

%=======================
\section{Preliminaries}
\label{section:preliminaries}
%=======================
%
We strictly adhere to the notations, definitions, and results
introduced in~\cite{SpotoMP10}. We briefly list the
elements that are relevant to this paper.
% , copying them from~\cite{SpotoMP10}.

For ease of exposition, we consider a simplification of the
Java bytecode where values can only be integers, locations
or \texttt{null}. 
\begin{definition}
  \label{def:state}
  The set of \emph{values} is $\mathbb{Z}\cup
  \mathbb{L}\cup\{\mathtt{null}\}$, where 
  $\mathbb{Z}$ is the set of integers and
  $\mathbb{L}$ is the set of \emph{memory locations}.
  A \emph{state} of the Java Virtual Machine is a triple
  $\langle l\sep s\sep\mu\rangle$ where $l$ is an array of values,
  called \emph{local variables} and numbered from $0$ upwards,
  $s$ is a stack of values, called \emph{operand stack} (in the
  following, just \emph{stack}), which grows leftwards,
  and $\mu$ is a \emph{memory}, or \emph{heap}, which maps
  \emph{locations} into \emph{objects}.
  An object is a function that maps its fields into values.
  We write $l^k$ for the value of the $k$th local variable;
  we write $s^k$ for the value of the $k$th stack element
  ($s^0$ is the base of the stack, $s^1$ is the element above and so on).
  The set of all states is denoted by $\Sigma$. When we want to fix
  the exact number $\#l\in\nat$ of local variables and $\#s\in\nat$
  of stack elements allowed in a state, we write $\Sigma_{\#l,\#s}$.
  \qed
\end{definition}
\begin{example}\label{example:state}
  Consider a memory
  \[\mu = [\ell_1\mapsto o_1,\ell_2\mapsto o_2,\ell_3\mapsto o_3,
  \ell_4\mapsto o_4,\ell_5\mapsto o_5]\]
  where $o_1=[f\mapsto\ell_4]$, $o_2=[f\mapsto\mathtt{null}]$,
  $o_3=[f\mapsto\ell_5]$, $o_4=[f\mapsto\mathtt{null}]$ and
  $o_5=[f\mapsto\mathtt{null}]$. Then,
  \[\sigma=\langle[5,\ell_2]\sep \ell_1::\ell_2::\ell_3\sep\mu\rangle\]
  is a state in $\Sigma_{2,3}$. Here, $\ell_1$ is the topmost element
  of the stack of $\sigma$, $\ell_2$ is the underlying element and
  $\ell_3$ is the element still below it. \qed
\end{example}
\begin{definition}
  \label{def:types}
  The set of \emph{types} of our simplified Java Virtual Machine
  is $\mathbb{T}=\mathbb{K}\cup\{\mathtt{int},\mathtt{void}\}$,
  where $\mathbb{K}$ is the set of all classes.
The $\mathtt{void}$ type can only be used as the return type of methods.
  A method signature is denoted by $\kappa.m(t_1,\ldots,t_p):t$ standing
  for a method named $m$, defined in class $\kappa$,
  expecting $p$ explicit parameters of type, respectively,
  $t_1,\ldots,t_p$ and returning a value of type $t$, or returning no
  value when $t=\mathtt{void}$.\qed
\end{definition}
We recall that in object-oriented languages, a non-static method
$\kappa.m(t_1,\ldots,t_p):t$ has also an \emph{implicit} parameter of
type $\kappa$ called \texttt{this} inside the code of the method.
Hence, the actual number of parameters is $p+1$.

A restricted set of eleven Java bytecode instructions is considered
in~\cite{SpotoMP10}. These instructions exemplify the operations
that the Java Virtual Machine performs.
Similarly, in this paper we only consider nine instructions,
but our implementation handles most of their variants.
%
%%% Definition: correct instructions
\begin{definition}\label{def:instructions}
  We let $\correctins$ denote the set consisting of the 
  following Java bytecode instructions.
  \begin{itemize}
  \item $\mathsf{const}\ c$, pushes the constant $c$ on top
    of the stack.
  \item $\mathsf{dup}$, duplicates the topmost element of the stack.
  \item $\mathsf{new}\ \kappa$, creates an object of class $\kappa$
    and pushes a reference to it on the stack.
  \item $\mathsf{load}\ i$, pushes the value of local variable $i$ on
    top of the stack.
  \item $\mathsf{store}\ i$, pops the top value from the stack and
    writes it into local variable $i$.
  \item $\mathsf{add}$, pops the topmost two values from the stack and
    pushes their sum instead.
  \item $\mathsf{putfield}\ f$, where $f$ has integer type, 
    pops the topmost two values $v$ (the top)
    and $\ell$ (under $v$) from the stack where $\ell$ must be a
    reference to an object $o$ or $\mathtt{null}$;
    if $\ell$ is $\mathtt{null}$, the computation stops,
    else $v$ is stored into field $f$ of $o$.
  \item $\mathsf{ifeq\ of\ type}\ t$, with 
    $t\in\mathbb{K}\cup\{\mathtt{int}\}$,
    pops the topmost element from the
    stack and checks if it is 0 (when $t$ is $\mathtt{int}$) or
    $\mathtt{null}$ (when $t$ is a class); if it is not the case,
    the computation stops.
  \item $\mathsf{if}\mathit{\langle{}cond\rangle}
    \mathsf{\ of\ type}\ \mathtt{int}$, with 
    $\mathit{cond}\in\{\mathsf{lt},\mathsf{le},
    \mathsf{gt},\mathsf{ge}\}$,
    pops the topmost element from the stack and checks,
    respectively, if it is less than 0,
    less than or equal to 0, greater than 0,
    greater than or equal to 0;
    if it is not the case, the computation stops.
  \item $\mathsf{call}\ \kappa.m(t_1,\ldots,t_p):t$.
    If $m$ is a static method, this instruction 
    pops the topmost $p$ values (the \emph{actual parameters})
    $a_1$, \ldots, $a_p$ from the stack (where $a_p$ is the
    topmost value) and $m$ is run from a state having an empty stack
    and a set of local variables bound to $a_1,\ldots,a_p$.
    If $m$ is not static, this instruction
    pops the topmost $p+1$ values (the \emph{actual parameters})
    $a_0$, $a_1$, \ldots, $a_p$ from the stack (where $a_p$ is the
    topmost value).
    Value $a_0$ is called \emph{receiver} of the call and must be
    $\mathtt{null}$ or a reference to an object of class $\kappa$
    or of a subclass of $\kappa$.
    If the receiver is $\mathtt{null}$, the computation stops.
%     Otherwise, a lookup procedure is started from the class $\kappa$ of $o$
%     upwards along the superclass chain, looking for a method called $m$,
%     expecting $p$ \emph{formal parameters}
%     of type $t_1,\ldots,t_p$, respectively, and returning a
%     value
%     of type $t$. It is guaranteed that such a method is found in a
%     class belonging to the set $\{\kappa_1,\ldots,\kappa_n\}$.
    Otherwise, method $m$ is run from a state having an empty stack
    and a set of local variables bound to $a_0,a_1,\ldots,a_p$.\qed
  \end{itemize}
\end{definition}
Unlike~\cite{SpotoMP10}, we do not consider the instruction
$\mathsf{getfield}\ f$, which is used for getting the value
of the field $f$ of an object, and
$\mathsf{putfield}\ f$, where $f$ has class type.
This is because we cannot design an exact abstraction,
as defined in Sect.~\ref{section:exact-approximations},
of these instructions.
We also do not consider the instruction
$\mathsf{ifne\ of\ type}\ t$, which pops the topmost element from
the stack, checks if it is 0 (when $t$ is \texttt{int}) or
\texttt{null} (when $t$ is a class) and, if it is the case,
stops the computation. This is because we have implemented
the results of this paper inside the Julia analyser,
which now systematically replaces the $\mathsf{ifne}$ instruction
with a disjunction of $\mathsf{iflt}$ (less than 0) and
$\mathsf{ifgt}$ (greater than 0); these two instructions belong
to the set considered by our implementation.
Finally, the $\mathsf{call}$ instruction considered
in~\cite{SpotoMP10} has the form 
\[\mathsf{call}\ \kappa_1.m(t_1,\ldots,t_p):t,\ldots,
\kappa_n.m(t_1,\ldots,t_p):t\] where 
$S=\{\kappa_1.m(t_1,\ldots,t_p):t,\ldots,\kappa_n.m(t_1,\ldots,t_p):t\}$
is an over-approximation of the set of methods that might be
called at run-time, at the program point where the call
occurs. This is because object-oriented languages,
such as Java bytecode, allow dynamic lookup of method
implementations in method calls, on the basis of the run-time
class of their receiver. Hence, the exact control-flow graph
of a program is not computable in general, but an
over-approximation can be computed instead.
In this paper, we present a technique for proving
\emph{existential} non-termination \ie for proving that
\emph{there exists} some inputs that lead to an
infinite execution. So, we have to ensure that the
methods we consider in the $\mathsf{call}$ instructions
are effectively called at run-time: this happens when $S$
only consists of one element.
Therefore, unlike~\cite{SpotoMP10}, we only consider 
calls of the form $\mathsf{call}\ \kappa.m(t_1,\ldots,t_p):t$
in this paper, and our technique cannot deal with situations
where $S$ consists of more than one element.

We assume that \emph{flat} code, as the one in Fig.~\ref{fig:sum_term_java},
is given a structure in terms of blocks of code linked by arrows expressing how
the flow of control passes from one to another.
% These might be for instance the \emph{basic blocks} of~\cite{AhoSU86},
We require that a $\mathsf{call}$ instruction can only occur at the
beginning of a block. For instance, Fig.~\ref{fig:sum_nonterm_blocks}
shows the blocks derived from the code of the method \texttt{sum} in
Fig.~\ref{fig:sum_term_java}. Note that at the beginning of the methods,
the local variables hold the parameters of the method.

From now on, a \emph{Java bytecode program} will be a graph of blocks, such
as that in Fig.~\ref{fig:sum_nonterm_blocks}; inside each block, there is
one or more instructions among those described in Definition~\ref{def:instructions}.
This graph typically contains many disjoint subgraphs, each corresponding to
a different method or constructor. The ends of a method or constructor, where
the control flow returns to the caller, are the end of every block with no
successor, such as the leftmost one in Fig.~\ref{fig:sum_nonterm_blocks}.
For simplicity, we assume that the stack there contains exactly as many elements
as are needed to hold the return value (normally $1$ element, but
$0$ element in the case of methods returning $\mathtt{void}$, such as all the
constructors or the \texttt{main} method).

A denotational semantics for Java bytecode is presented in~\cite{SpotoMP10}
together with a \emph{path-length} relational abstract domain that is
used for proving termination of Java bytecode programs. Denotations
are state transformers that can be composed to model the sequential
execution of instructions.
\begin{definition}
  \label{def:denotation}
  A denotation is a partial function $\Sigma\to\Sigma$
  from an \emph{input} state to an \emph{output} or \emph{final}
  state. The set of denotations is denoted by $\Delta$.
  When we want to fix the number of local variables and stack elements
  in the input and output states, we write $\Delta_{l_i,s_i\to l_o,s_o}$,
  standing for $\Sigma_{l_i,s_i}\to\Sigma_{l_o,s_o}$.
  Let $\delta_1,\delta_2\in\Delta$.
  Their \emph{sequential composition} is
  $\delta_1;\delta_2=\lambda\sigma.\delta_2(\delta_1(\sigma))$,
  which is undefined when $\delta_1(\sigma)$ is undefined
  or when $\delta_2(\delta_1(\sigma))$ is undefined. \qed
\end{definition}
For each instruction $\mathsf{ins}$ in $\correctins$ and
program point $q$ where $\mathsf{ins}$ occurs,
\cite{SpotoMP10} provides the definition of a corresponding
denotation $\mathit{ins}_q$.
\begin{example}
  Let $q$ be a program point where the instruction $\mathsf{dup}$
  occurs and let $\#l$ and $\#s$ be the number of local variables
  and stack elements at $q$. The denotation $\mathit{dup}_q$ corresponding
  to $\mathsf{dup}$ at $q$ is defined as: 
  $\mathit{dup}_q\in \Delta_{\#l,\#s\to \#l,\#s+1}$ and
  $\mathit{dup}_q = \lambda\langle l\sep\mathit{top}::s\sep\mu\rangle.
  \langle l\sep\mathit{top}::\mathit{top}::s\sep\mu\rangle$ where
  $\mathit{top}::s$ denotes a non-empty stack whose top element is
  $\mathit{top}$ and remaining portion is $s$. \qed
\end{example}
\cite{SpotoMP10} also defines the abstraction of
$\mathit{ins}_q$ into its path-length polyhedron
$\mathit{ins}^{\mathbb{PL}}_q$.
\begin{definition}
  \label{def:path_length_domain}
  Let $l_i,s_i,l_o,s_o\in\nat$.
  The set $\mathbb{PL}_{l_i,s_i\to l_o,s_o}$
  of the \emph{path-length polyhedra} contains all finite
  sets of integer linear constraints over the variables
  $\{\inputv{l}^k\mid 0\le k <l_i\}\cup\{\inputv{s}^k\mid 0\le k<s_i\}
  \cup\{\outputv{l}^k\mid 0\le k<l_o\}\cup\{\outputv{s}^k\mid 0\le k<s_o\}$,
  using only the $\le$, $\ge$ and $=$ comparison operators. \qed
\end{definition}
The path-length polyhedron $\mathit{ins}_q^{\mathbb{PL}}$ describes the
relationship between the \emph{sizes} $\inputv{l}^k$ and $\inputv{s}^k$
of the local variables and stack elements in the input state
of $\mathit{ins}_q$ and the sizes $\outputv{l}^k$ and $\outputv{s}^k$
of the local variables and stack elements in the output state
of $\mathit{ins}_q$.
The size of a local variable or stack element $v$ in a 
memory $\mu$ is denoted by $\len(v,\mu)$ and is informally
defined as: if $v\in\mathbb{Z}$ then $\len(v,\mu)=v$,
if $v$ is \texttt{null} then $\len(v,\mu)=0$ and
if $v$ is a location then $\len(v,\mu)$ is the maximal
length in $\mu$ of a chain of locations that one can follow
from $v$.
\begin{example}
  \cite{SpotoMP10} defines
  $\mathit{dup}_q^\mathbb{PL}=\mathit{Unchanged}_q(\#l,\#s)\cup
  \{\inputv{s}^{\#s-1}=\outputv{s}^{\#s}\}$
  where $\mathit{Unchanged}_q(\#l,\#s) =
  \{\inputv{l}^i=\outputv{l}^i\mid 0\leq i < \#l\}
  \cup\{\inputv{s}^i=\outputv{s}^i\mid 0\leq i < \#s\}$.
  Hence, $\mathit{dup}_q^\mathbb{PL}$ expresses the fact that
  after an execution of $\mathsf{dup}$, the new top of the
  stack has the same path-length as the former one
  ($\{\inputv{s}^{\#s-1}=\outputv{s}^{\#s}\}$)
  and that the path-length of the local variables and stack
  elements is unchanged ($\mathit{Unchanged}_q(\#l,\#s)$).
  \qed
\end{example}

Note that \cite{SpotoMP10} also provides the definition
of the abstract counterpart $;^{\mathbb{PL}}$ of the 
operator $;$ used for composing denotations. The operator
$;^{\mathbb{PL}}$ is hence used for composing path-length
polyhedra.
\begin{definition}
  \label{def:abstract_transformers}
  Let $\pl_1\in\mathbb{PL}_{l_i,s_i\to l_t,s_t}$ together with
  $\pl_2\in\mathbb{PL}_{l_t,s_t\to l_o,s_o}$.
  Let $T=\{\overline{l}^0,\ldots,\overline{l}^{l_t-1},
  \overline{s}^0,\ldots,\overline{s}^{s_t-1}\}$. We define
  $\pl_1;^\mathbb{PL}\pl_2\in\mathbb{PL}_{l_i,s_i\to l_o,s_o}$ as
  \[
  \pl_1;^\mathbb{PL}\pl_2=\exists_T\left(
    \pl_1[\outputv{v}\mapsto\overline{v}\mid\overline{v}\in T]\cup
    \pl_2[\inputv{v}\mapsto\overline{v}\mid\overline{v}\in T]
  \right)
  \]
  where $\pl_1[\outputv{v}\mapsto\overline{v}\mid\overline{v}\in T]$
  (resp. $\pl_2[\inputv{v}\mapsto\overline{v}\mid\overline{v}\in T]$)
  denotes the replacement in $\pl_1$ (resp. $\pl_2$) of
  $\outputv{v}$ with $\overline{v}$
  (resp. $\inputv{v}$ with $\overline{v}$).
  \qed
\end{definition}

The abstractions $\mathit{ins}_q^{\mathbb{PL}}$,
for each $\mathsf{ins}\in\correctins$, and
the abstraction $;^{\mathbb{PL}}$ are all proved to be
\emph{correct} \ie \cite{SpotoMP10} provides the proof
that these abstractions include their concrete counterpart
in their concretisation.

\begin{definition}
  \label{def:model}
  Let $\mathit{pl}\in\mathbb{PL}_{l_i,s_i\to l_o,s_o}$ and
  $\mathit{\rho}$ be an assignment from a superset of the
  variables of $\mathit{pl}$ into $\mathbb{Z}\cup\{+\infty\}$.
  We say that $\rho$ is a \emph{model} of $\mathit{pl}$
  and we write $\rho\models\mathit{pl}$ when $\rho(\mathit{pl})$
  is true, that is, by substituting, in $\pl$,
  the variables with their values provided by $\rho$, we
  get a tautological set of ground constraints.
  \qed
\end{definition}

Any state can be mapped into an input path-length assignment,
when it is considered as the input state of a denotation,
or into an output path-length assignment, when it is
considered as the output state of a denotation.
\begin{definition}
  \label{def:assignment}
  Let $\langle l\sep s\sep\mu\rangle\in\Sigma_{\#l,\#s}$. Its
  \emph{input path-length assignment} is
  \begin{align*}
    \inputv{\len}(\langle l\sep s\sep\mu\rangle) = & 
    \;[\inputv{l}^k\mapsto\len(l^k,\mu)\mid 0\le k <\#l]\\
    & \cup[\inputv{s}^k\mapsto\len(s^k,\mu)\mid 0\le k <\#s]
  \end{align*}
  and, similarly, its \emph{output path-length assignment} is
  \begin{align*}
    \outputv{\len}(\langle l\sep s\sep\mu\rangle) = &
    \;[\outputv{l}^k\mapsto\len(l^k,\mu)\mid 0\le k <\#l]\\
    & \cup [\outputv{s}^k\mapsto\len(s^k,\mu)\mid 0\le k <\#s]~.
  \end{align*}
  \qed
\end{definition}
\begin{example}
  \label{example:len}
  In Example~\ref{example:state},
  \begin{align*}
    \inputv{\len}(\sigma)&=\left[\begin{array}{l}
        \inputv{l}^0\mapsto \len(5,\mu),\\
        \inputv{l}^1\mapsto \len(\ell_2,\mu),\\
        \inputv{s}^0 \mapsto \len(\ell_3,\mu),\\
        \inputv{s}^1 \mapsto \len(\ell_2,\mu),\\
        \inputv{s}^2 \mapsto \len(\ell_1,\mu)
      \end{array}\right]
    =\left[\begin{array}{l}
        \inputv{l}^0\mapsto 5,\\
        \inputv{l}^1\mapsto 1,\\
        \inputv{s}^0 \mapsto 2,\\
        \inputv{s}^1 \mapsto 1,\\
        \inputv{s}^2 \mapsto 2
      \end{array}\right]\;.
  \end{align*}
  Similarly, 
  \[\outputv{\len}(\sigma) =
  [\outputv{l}^0\mapsto 5, \outputv{l}^1\mapsto 1,
  \outputv{s}^0 \mapsto 2, \outputv{s}^1 \mapsto 1,
  \outputv{s}^2 \mapsto 2]\;.\] \qed
\end{example}

%%%%%%%%%%%%%%%%%%%%%%%%%%%%%%%%%%%%%%%%%%%%%%%%%%%%%%%%%%%%%%%%%

%==============================
\section{Exact Abstractions}
\label{section:exact-approximations}
%==============================
%
Our technique for proving non-termination of a
Java bytecode program $P$ consists in abstracting
$P$ as a $\clppl$ program $P_{\mathit{CLP}}$, then
in proving non-termination of $P_{\mathit{CLP}}$,
and finally in concluding the non-termination of $P$
from that of $P_{\mathit{CLP}}$, when it is possible.
In~\cite{Payet09b}, we observed informally that when the abstraction
of $P$ as $P_{\mathit{CLP}}$ is \emph{exact}, the non-termination
of $P_{\mathit{CLP}}$ entails that of $P$.
In this section, we give a formal definition of exactness.
First, we start with preliminary definitions, where we let
$\codom(\delta_1)$ denote the codomain of the denotation $\delta_1$.
%
%%% Definition: input and output parts of an assignment
\begin{definition}\label{definition:input_output_assignment}
  Let $\pl\in\mathbb{PL}_{l_i,s_i\to l_o,s_o}$ and $\rho$
  be a model of $\pl$.
  We let $\inputv{\rho}$ denote the assignment obtained by
  restricting the domain of $\rho$ to the input variables
  $\inputv{l}^0$, \ldots, $\inputv{l}^{l_i-1}$
  and $\inputv{s}^0$, \ldots, $\inputv{s}^{s_i-1}$.
  We let $\outputv{\rho}$ denote the assignment obtained by
  restricting the domain of $\rho$ to the output variables
  $\outputv{l}^0$, \ldots, $\outputv{l}^{l_o-1}$
  and $\outputv{s}^0$, \ldots, $\outputv{s}^{s_o-1}$.  
  \qed
\end{definition}
%
%%% Definition: compatibility
\begin{definition}
  We say that a state $\sigma$ is \emph{compatible} with a
  denotation $\delta$ when $\sigma$ satisfies the static
  information at $\delta$ (number and type of local variables
  and stack elements). We say that a denotation $\delta_1$
  is \emph{compatible} with a denotation $\delta_2$ when
  any state in $\codom(\delta_1)$ is compatible
  with $\delta_2$. \qed
\end{definition}

Our definition of exactness is the following. Intuitively,
the abstraction of a denotation $\delta$ into a path-length
polyhedron $\pl$ is exact when $\pl$, considered as an 
input-output mapping from input to output variables,
exactly matches $\delta$ \ie any model of $\pl$ only
corresponds to states for which $\delta$ is defined.
%
%%% Definition: exactness
\begin{definition}
  \label{definition:models}
  Let $\delta\in\Delta_{l_i,s_i\to l_o,s_o}$ and
  $\pl\in\mathbb{PL}_{l_i,s_i\to l_o,s_o}$.
  We say that $\pl$ is an \emph{exact abstraction} of $\delta$,
  and we write $\pl \models \delta$, when for any model $\rho$
  of $\pl$ and any state $\sigma$ compatible with $\delta$,
  $\inputv{\len}(\sigma) = \inputv{\rho}$ implies
  that $\delta(\sigma)$ is defined and
  $\outputv{len}(\delta(\sigma)) = \outputv{\rho}$.
  \qed
  %
%   We say that a bytecode instruction $\mathsf{ins}$ is
%   \emph{exactly translated} when
%   the corresponding denotation
%   $\mathit{ins}$ and path-length polyhedron $\mathit{ins}^{\mathbb{PL}}$
%   given in~\cite{SpotoMP10} satisfy
%   $\mathit{ins}_q^{\mathbb{PL}}\models\mathit{ins}_q$
%   for any program point $q$.
  % 
%   We say that a Java bytecode program is exactly translated
%   when it consists of instructions that are exactly translated.
\end{definition}

Exactness is preserved by sequential composition:
%
%%% Proposition: composition
\begin{proposition}\label{proposition:composition}
  Let $\delta_1\in\Delta_{l_i,s_i\to l_t,s_t}$,
  $\pl_1\in\mathbb{PL}_{l_i,s_i\to l_t,s_t}$ be such that
  $\pl_1\models\delta_1$.
  Let $\delta_2\in\Delta_{l_t,s_t\to l_o,s_o}$ and
  $\pl_2\in\mathbb{PL}_{l_t,s_t\to l_o,s_o}$ be such that
  $\pl_2\models\delta_2$.
  Suppose that $\delta_1$ is compatible with $\delta_2$.
  Then, we have $\pl_1;^{\mathbb{PL}}\pl_2\models\delta_1;\delta_2$.
\end{proposition}

Except from $\mathsf{call}$, all the bytecode instructions we consider
in this paper are exactly abstracted:
%
%%% Proposition: correct instructions
\begin{proposition}\label{proposition:instructions_exactness}
  For any $\mathsf{ins} \in \correctins\setminus 
  \{\mathsf{call}\}$ and program point $q$ where 
  $\mathsf{ins}$ occurs,
  we have $\mathit{ins}_q^{\mathbb{PL}}\models\mathit{ins}_q$.
\end{proposition}
The proof for % the instruction
$\mathsf{new}\ \kappa$ assumes that the
denotation of this bytecode is a total map. This is true only if we
assume that the system has infinite memory.
Termination caused by out of memory is not really termination from
our point of view. 
We deal with the $\mathsf{call}$
instruction in Sect.~\ref{section:compilation}: we do not abstract its
denotation into a path-length polyhedron but rather
translate it into a call to a predicate
(see Definitions~\ref{def:clp_block_call_not_void}%
--\ref{def:clp_block_call_not_void_one_inst} below).

\begin{example}
  Let $q$ be a program point where the instruction $\mathsf{dup}$
  occurs and let $\#l$ and $\#s$ be the number of local variables
  and stack elements at $q$. We have 
  \begin{align*}
    \mathit{dup}_q
    &= \lambda\langle l\sep\mathit{top}::s\sep\mu\rangle.
    \langle l\sep\mathit{top}::\mathit{top}::s\sep\mu\rangle\\
    \mathit{dup}_q^\mathbb{PL}
    &=\mathit{Unchanged}_q(\#l,\#s)\cup
    \{\inputv{s}^{\#s-1}=\outputv{s}^{\#s}\}\;.
  \end{align*}
  Let $\rho$ be a model of $\mathit{dup}_q^\mathbb{PL}$.
  Let $\sigma$ be a state that is compatible with
  $\mathit{dup}_q$. Then, $\sigma\in\Sigma_{\#l,\#s}$ and
  $\sigma$ has the form $\langle l\sep\mathit{top}::s\sep\mu\rangle$.
  Clearly, $\mathit{dup}_q(\sigma)$ is defined and we have
  $\mathit{dup}_q(\sigma) =
  \langle l\sep\mathit{top}::\mathit{top}::s\sep\mu\rangle$.
  Suppose that $\inputv{\len}(\sigma) = \inputv{\rho}$.
  \begin{itemize}
  \item For any $l^i\in l$, we have
    $\outputv{len}(\mathit{dup}_q(\sigma))(\outputv{l}^i) = 
    \inputv{\len}(\sigma)(\inputv{l}^i)$ because $\sigma$ and
    $\mathit{dup}_q(\sigma)$ have the same array of local variables
    $l$ and memory $\mu$.
    Moreover, as $\inputv{\len}(\sigma) = \inputv{\rho}$, we have
    $\inputv{\len}(\sigma)(\inputv{l}^i) = \inputv{\rho}(\inputv{l}^i)$.
    As $\rho$ is a model of $\mathit{dup}_q^\mathbb{PL}$, with
    $\mathit{Unchanged}_q(\#l,\#s)\subseteq\mathit{dup}_q^\mathbb{PL}$,
    we have $\inputv{\rho}(\inputv{l}^i) = \outputv{\rho}(\outputv{l}^i)$.
    Therefore, $\outputv{len}(\mathit{dup}_q(\sigma))(\outputv{l}^i)
    = \outputv{\rho}(\outputv{l}^i)$.
  \item Similarly, for any $s^i\in \mathit{top}::s$, % we have
    $\outputv{len}(\mathit{dup}_q(\sigma))(\outputv{s}^i)
    = \outputv{\rho}(\outputv{s}^i)$.
  \item Finally, consider the
    top element $s^{\#s} = \mathit{top}$ in the stack of 
    $\mathit{dup}_q(\sigma)$. We have
    \begin{align*}
      \outputv{len}(\mathit{dup}_q(\sigma))(\outputv{s}^{\#s})
      &= \len(\mathit{top},\mu)\\
      &= \inputv{len}(\sigma)(\inputv{s}^{\#s-1}) =
      \inputv{\rho}(\inputv{s}^{\#s-1})\;.
    \end{align*}
    As $\rho$ is a model of $\mathit{dup}_q^\mathbb{PL}$, with
    $\{\inputv{s}^{\#s-1}=\outputv{s}^{\#s}\}\subseteq
    \mathit{dup}_q^\mathbb{PL}$, we have
    $\inputv{\rho}(\inputv{s}^{\#s-1})=
    \outputv{\rho}(\outputv{s}^{\#s})$. So, 
    $\outputv{len}(\mathit{dup}_q(\sigma))(\outputv{s}^{\#s})
    =\outputv{\rho}(\outputv{s}^{\#s})$.
  \end{itemize}
  Consequently, we have 
  $\outputv{len}(\mathit{dup}_q(\sigma)) = \outputv{\rho}$.
  \qed
\end{example}

% An important step of our analysis consists in deducing the
% non-termination of $P$ from that of $P_{\mathit{CLP}}$.
% The idea consists in constructing an infinite execution of $P$
% from an infinite derivation with $P_{\mathit{CLP}}$. Each step of
% the infinite derivation consists of an atom $b(\mathit{vars})$,
% where $\mathit{vars}$ are integer values, and we must be able to
% transform this atom into a state whose path-length matches
% $\mathit{vars}$.
% %
% %%% Proposition:
% \begin{proposition}\label{proposition:instructions_exists_sigma}
%   For any
% %  $\mathsf{ins} \in \correctins \setminus \{\mathsf{putfield}\ f\}$ 
%   $\mathsf{ins} \in \correctins\setminus\{\mathsf{call}\}$ 
%   and any model $\rho$ of $\mathit{ins}^{\mathbb{PL}}$,
%   there exists a state $\sigma$ which is compatible with
%   $\mathit{ins}$ and such that 
%   $\inputv{\len}(\sigma) = \inputv{\rho}$.
% \end{proposition}

%%%%%%%%%%%%%%%%%%%%%%%%%%%%%%%%%%%%%%%%%%%%%%%%%%%%%%%%%%%%%%%%%

%====================================================
\section{Constraint Logic Program derived from
  Java bytecode}
\label{section:compilation}
%====================================================
%
The technique that we describe in~\cite{SpotoMP10} for proving
the termination of a Java bytecode program $P$ computes a $\clppl$
program $P_{\mathit{CLP}}$ which is an \emph{over-approximation} of
$P$, in the sense that the set of executions of $P$ is ``included''
in that of $P_{\mathit{CLP}}$. This is because some bytecode
instructions considered in~\cite{SpotoMP10}
\eg $\mathsf{call}$, are not exactly
abstracted, in the sense of Definition~\ref{definition:models},
but are over-approximated instead.
\begin{example}
  The bytecode instruction $\mathsf{getfield}\ f$ pops the
  topmost value $\ell$ of the stack, which must be a
  reference to an object $o$ or $\mathtt{null}$, and pushes
  $o(f)$ at its place. If $\ell$ is $\mathtt{null}$, the
  computation stops.
  For any program point $q$ with $\#l$ local variables
  and $\#s$ stack elements and any field $f$ with integer type,
  \cite{SpotoMP10} defines:
  \begin{align*}
    \mathit{getfield}_q\ f
    &= \lambda\langle l\sep\ell::s\sep\mu\rangle.\\
    & \hspace{1.6cm}\begin{cases}
      \langle l\sep\mu(\ell)(f)::s\sep\mu\rangle & \text{if $\ell\not=\mathtt{null}$}\\
      \mathit{undefined} & \text{otherwise}\\
    \end{cases}\\
    \mathit{getfield}_q^\mathbb{PL}\ f 
    &= \mathit{Unchanged}_q(\#l,\#s-1)\cup
    \{\inputv{s}^{\#s-1}\geq 1\}\;.
  \end{align*}
  The denotation $\mathit{getfield}_q\ f$ is not
  exactly abstracted by the path-length polyhedron
  $\mathit{getfield}_q^\mathbb{PL}\ f$ because 
  $\mathit{getfield}_q^\mathbb{PL}\ f$ does not provide
  any constraint for the top $\outputv{s}^{\#s-1}$ of the
  output stack, while in $\mathit{getfield}_q\ f$ a new element 
  appears on top of the output stack.
  Hence, $\mathit{getfield}_q^\mathbb{PL}\ f$ does not
  exactly matches $\mathit{getfield}_q\ f$:
  for some model $\rho$ of $\mathit{getfield}_q^\mathbb{PL}\ f$,
  there exists a state $\sigma$ which is such that 
  $\inputv{\len}(\sigma) = \inputv{\rho}$ and
  $\outputv{len}(\delta(\sigma)) \neq \outputv{\rho}$.
  For instance, suppose that $\#l = 2$ and $\#s = 3$
  and let
  \begin{align*}
    \rho = [
    & \inputv{l}^0\mapsto 5,
    \inputv{l}^1\mapsto 1,
    \inputv{s}^0\mapsto 2,
    \inputv{s}^1\mapsto 1,
    \inputv{s}^2\mapsto 2,\\
    & \outputv{l}^0\mapsto 5,
    \outputv{l}^1\mapsto 1,
    \outputv{s}^0\mapsto 2,
    \outputv{s}^1\mapsto 1,
    \outputv{s}^2\mapsto 10] \;.
  \end{align*}
  Then, $\rho$ is a model of $\mathit{getfield}_q^\mathbb{PL}\ f$.
  The state
  \[\sigma=\langle[5,\ell_2]\sep \ell_1::\ell_2::\ell_3\sep\mu\rangle\]
  of Example~\ref{example:state} is compatible with 
  $\mathit{getfield}_q\ f$ and, by Example~\ref{example:len},
  we have $\inputv{\len}(\sigma) = \inputv{\rho}$. Moreover,
  as $\mu(\ell_1)(f) = \ell_4$,
  \[(\mathit{getfield}_q\ f)(\sigma) =
  \langle[5,\ell_2]\sep \ell_4::\ell_2::\ell_3\sep\mu\rangle\;.\]
  As $\len(\ell_4,\mu) = 1$, we have
  \[\begin{array}{l}
    \outputv{\len}((\mathit{getfield}_q\ f)(\sigma)) = \\[1ex]
    \multicolumn{1}{r}{\hspace{2cm}
      [\outputv{l}^0\mapsto 5,
      \outputv{l}^1\mapsto 1,
      \outputv{s}^0\mapsto 2,
      \outputv{s}^1\mapsto 1,
      \outputv{s}^2\mapsto 1]\;.}
  \end{array}\]
  Consequently, 
  $\outputv{\len}((\mathit{getfield}_q\ f)(\sigma))
  \neq \outputv{\rho}$. \qed
\end{example}

For non-termination,
we rather need an \emph{under-approximation} of $P$, \ie a
program whose set of executions is ``included'' in that of $P$.
Note that because of
Proposition~\ref{proposition:instructions_exactness},
when $P$ only consists of instructions in
$\correctins\setminus\{\mathsf{call}\}$, the set of executions
of $P_{\mathit{CLP}}$, computed as in~\cite{SpotoMP10}, exactly matches
that of $P$ and we have:
%
%%% Theorem: termination
\begin{theorem}
  \label{theorem:clp_completeness}
  Let $J$ be a Java Virtual Machine,
  $P$ be a Java bytecode program consisting of
  instructions in $\correctins\setminus\{\mathsf{call}\}$, and
  $b$ be a block of $P$. 
  Let $P_{\mathit{CLP}}$ be the abstraction of $P$ as a
  $\clppl$ program, computed as in~\cite{SpotoMP10}.
  The query $b(\mathit{vars})$ has only terminating computations
  in $P_{\mathit{CLP}}$, for any fixed integer values for
  $\mathit{vars}$, \emph{if and only if} all executions of $J$
  started at block $b$ terminate.
\end{theorem}

In this section, we consider a Java bytecode program $P$
consisting of \emph{any} instructions in $\correctins$
(including $\mathsf{call}$) and we describe a technique for
abstracting $P$ as a $\clppl$ program $P_{\mathit{CLP}}$ whose
non-termination entails that of $P$. We do not abstract the
$\mathsf{call}$ instruction into a path-length polyhedron
but we rather translate it into an explicit call
to a predicate. We consider a block $b$:
\[\block{\begin{array}{c|c}b & \begin{smallmatrix}
      \mathsf{ins}_1\\
      \mathsf{ins}_2\\
      \cdots\\
      \mathsf{ins}_w
    \end{smallmatrix}\end{array}}\ttob\]
of $P$ occurring in a method $m_b$ and we describe the set
$b_{\mathit{CLP}}$ of $\clppl$ clauses derived from $b$; here,
$\mathsf{ins}_1$ $\mathsf{ins}_2$, \ldots, $\mathsf{ins}_w$ are
the instructions of $b$ and $b_1$, \ldots, $b_n$ are the successor
blocks of $b$ in $P$.
We let $\inputv{\vars}$ be the input local variables and stack elements at
the beginning of $b$ and $\outputv{\vars}$ be the output local variables
and stack elements at the end of $b$ (in some fixed order).

We let $\inputv{s}_b$ and $\outputv{s}_b$ be some fresh variables,
not occurring in $\inputv{\vars} \cup \outputv{\vars}$, and we use them
to capture the path-length of the return value of $m_b$.
In Definitions~\ref{def:clp_block_no_call_not_void}--%
\ref{def:clp_block_call_not_void_one_inst} below, when $b$ has no
successor ($n=0$), at the end of $b$ the stack contains exactly the return
value of $m_b$; hence, $\outputv{s}^0$ is bound to the path-length of this
return value and we set $\inputv{s}_b=\outputv{s}^0$ in order to capture
this path-length.
It is important to remark that we assume a specialised semantics of
\textit{CLP} computations here, where variables are always bound to integer
values, except for $\inputv{s}_b$ and $\outputv{s}_b$.
This means that we do not allow \emph{free} variables in a call to a predicate,
except for $\inputv{s}_b$ and $\outputv{s}_b$ which are always free until they
get bound to a value in a clause corresponding to a block with no successor
(see the constraints $\inputv{s}_b=\outputv{s}^0$ in
Definitions~\ref{def:clp_block_no_call_not_void}--%
\ref{def:clp_block_call_not_void_one_inst} below).

First, we consider the situation where $b$ does not start with a
$\mathsf{call}$ instruction. For each successor $b_i$ of $b$,
we generate a clause of the form
\[b(\inputv{\vars},\inputv{s}_b) \la c\cup\{\inputv{s}_b=\outputv{s}_b\},
b_i(\outputv{\vars},\outputv{s}_b)\]
which indicates that the flow of control passes from $b$ to $b_i$. Here,
$c$ is a constraint which expresses the sequential execution of the
instructions of $b$.
\begin{definition}\label{def:clp_block_no_call_not_void}
  Suppose that $\mathsf{ins}_1$ is not a $\mathsf{call}$ instruction.
  Let
  \[c=\mathit{ins}_1^{\mathbb{PL}} ;^{\mathbb{PL}}\cdots
  ;^{\mathbb{PL}}\mathit{ins}_w^{\mathbb{PL}}\;.\]
  We define $b_{\mathit{CLP}}$ as follows.
  \begin{enumerate}
  \item If $n\neq 0$, $b_{\mathit{CLP}}$ is the set consisting 
    of the $\clppl$ clauses 
    \begin{equation*}
      \begin{array}{l}
        b(\inputv{\vars},\inputv{s}_b) \la 
        c\cup\{\inputv{s}_b=\outputv{s}_b\},b_1(\outputv{\vars},\outputv{s}_b)\\
        \cdots\\
        b(\inputv{\vars},\inputv{s}_b) \la c\cup\{\inputv{s}_b=\outputv{s}_b\}
        ,b_n(\outputv{\vars},\outputv{s}_b)
      \end{array}
    \end{equation*}
  \item If $n = 0$, $b_{\mathit{CLP}}$ is the set consisting of the
    $\clppl$ clause
    \begin{equation*}
      b(\inputv{\vars},\inputv{s}_b) \la c\cup\{\inputv{s}_b=\outputv{s}^0\}
      \;.
    \end{equation*}
  \end{enumerate}
  \qed
\end{definition}

Now, suppose that block $b$ starts with a $\mathsf{call}$
instruction to a non-static method $m$ with $p$ actual parameters.
Then, at the beginning of $b$, the actual parameters of $m$
sit on the top of the stack and, at the end of the
$\mathsf{call}$ instruction, these parameters
are replaced with the return value of $m$:
\[\underbrace{
  \langle l\sep \overbrace{a_{p-1}::\cdots::a_0}^{\text{actual parameters}}::
  s\sep\mu\rangle}_{\text{the state at the beginning of $b$}}
\quad\lra^{\mathsf{call}}\quad
\langle l\sep \overbrace{a}^{\text{return value}}::s\sep\mu'\rangle\]
($a_0$ is the receiver and the memory $\mu$ may be affected by the call).
Therefore, if $\#l$ and $\#s$ are the number of local variables
and stack elements after the $\mathsf{call}$ instruction, 
then, in the input state of $\mathsf{call}$,
$s^{\#s-1+p-1}$, \ldots, $s^{\#s-1}$ are the actual parameters
of $m$ where $s^{\#s-1}$ is the receiver
and, in the output state of $\mathsf{call}$,
$s^{\#s-1}$ is the return value of $m$.
Note that the array $l$ of local variables and the stack
portion $s$ under the actual parameters remain unchanged
after the call. In general, this does not mean that the
path-length of their elements remains the same, as 
the execution of a method may modify the memory $\mu$,
hence the path-length of locations in $l$ and $s$.
In the scope of this paper, however, we discard the
instructions of the form $\mathsf{putfield}\ f$ where $f$
has class type. Therefore, the instructions we consider
do not modify the path-length of locations; hence after
a method call, the path-length of the elements of $l$
and $s$ remains the same.
In Definition~\ref{def:clp_block_call_not_void}
and Definition~\ref{def:clp_block_call_not_void_one_inst}
below, this operational semantics of $\mathsf{call}$ is modelled by:
\begin{itemize}
\item the constraint 
  \[c_= = \{\inputv{l}^i=\outputv{l}^i\mid 0\leq i <\#l\}\cup
  \{\inputv{s}^i=\outputv{s}^i\mid 0\leq i <\#s-1\}\]
  which specifies that the path-length of the local variables
  and stack elements under the actual parameters is not modified
  by the call,
\item the constraint $\inputv{s}^{\#s-1}\geq 1$,
  which specifies that the receiver of the call is not 
  $\mathtt{null}$, and
\item the atom $b_m(\inputv{s}^{\#s-1+p-1},\ldots,\inputv{s}^{\#s-1},
  \outputv{s}^{\#s-1})$, where $b_m$ denotes the entry block of $m$.
\end{itemize}
When the call is over, control returns to the next instruction. We
distinguish two situations here: that where $b$ contains more than
one instruction (Definition~\ref{def:clp_block_call_not_void}), then
control returns to the instruction in $b$ following the call to $m$,
and the situation where $b$ consists of the call to $m$ only
(Definition~\ref{def:clp_block_call_not_void_one_inst}), then control
returns to a successor of $b$.
\begin{definition}\label{def:clp_block_call_not_void}
  Suppose that $\mathsf{ins}_1=\mathsf{call}\ m$ and that
  $\mathsf{ins}_1$ is not the only instruction in $b$ (\ie $w\geq 2$).
  Let
  \[c=\mathit{ins}_2^{\mathbb{PL}} ;^{\mathbb{PL}}\cdots
  ;^{\mathbb{PL}}\mathit{ins}_w^{\mathbb{PL}},\]
  let $\outputv{\vars}'$ be the output local variables and stack elements
  at the end of $\mathsf{ins}_1$ and
  $\inputv{\vars}'$ be the input local variables and stack elements
  at the beginning of $\mathsf{ins}_2$ (in some fixed order).
  We suppose that $\inputv{s}_b$ and $\outputv{s}_b$ do no
  occur in $\inputv{\vars}'$ and $\outputv{\vars}'$.
  We define $b_{\mathit{CLP}}$ as follows.
  \begin{enumerate}
  \item If $n \neq 0$, 
    $b_{\mathit{CLP}}$ is the set consisting 
    of the $\clppl$ clause 
    \begin{align*}
      b(\inputv{\vars},\inputv{s}_b) \la \; &
      c_= \cup \{\inputv{s}^{\#s-1}\geq 1,\ \inputv{s}_b=\outputv{s}_b\},\\
      & b_m(\inputv{s}^{\#s-1+p-1},\ldots,\inputv{s}^{\#s-1},
      \outputv{s}^{\#s-1}),\\
      & b'(\outputv{\vars}',\outputv{s}_b)
    \end{align*}
    together with
    \begin{align*}
      &b'(\inputv{\vars}',\inputv{s}_b) \la 
      c\cup\{\inputv{s}_b=\outputv{s}_b\},b_1(\outputv{\vars},\outputv{s}_b)\\
      &\cdots\\
      &b'(\inputv{\vars}',\inputv{s}_b) \la c\cup\{\inputv{s}_b=\outputv{s}_b\}
      ,b_n(\outputv{\vars},\outputv{s}_b)
    \end{align*}
    where $b'$ is a fresh predicate symbol.
  \item If $n = 0$,  
    $b_{\mathit{CLP}}$ is the set consisting 
    of the $\clppl$ clause
    \begin{align*}
        b(\inputv{\vars},\inputv{s}_b) \la \; &
        c_= \cup \{\inputv{s}^{\#s-1}\geq 1,\ \inputv{s}_b=\outputv{s}_b\},\\
        & b_m(\inputv{s}^{\#s-1+p-1},\ldots,\inputv{s}^{\#s-1},
        \outputv{s}^{\#s-1}),\\
        & b'(\outputv{\vars}',\outputv{s}_b)
    \end{align*}
    together with
    \[b'(\inputv{\vars}',\inputv{s}_b) \la 
    c\cup\{\inputv{s}_b=\outputv{s}^0\}\]
    where $b'$ is a fresh predicate symbol. \qed
  \end{enumerate}
\end{definition}

\begin{definition}\label{def:clp_block_call_not_void_one_inst}
  Suppose that $\mathsf{ins}_1=\mathsf{call}\ m$ and that
  $\mathsf{ins}_1$ is the only instruction in $b$ (\ie $w = 1$).
  We define $b_{\mathit{CLP}}$ as follows.
  \begin{enumerate}
  \item If $n \neq 0$, 
    $b_{\mathit{CLP}}$ is the set consisting 
    of the $\clppl$ clauses 
    \begin{align*}
      b(\inputv{\vars},\inputv{s}_b) \la \;
      & c_= \cup \{\inputv{s}^{\#s-1}\geq 1,\ \inputv{s}_b=\outputv{s}_b\},\\
      & b_m(\inputv{s}^{\#s-1+p-1},\ldots,\inputv{s}^{\#s-1},
      \outputv{s}^{\#s-1}),\\
      & b_1(\outputv{\vars},\outputv{s}_b)\\
      \cdots & \\
      b(\inputv{\vars},\inputv{s}_b) \la \;
      & c_= \cup \{\inputv{s}^{\#s-1}\geq 1,\ \inputv{s}_b=\outputv{s}_b\},\\
      & b_m(\inputv{s}^{\#s-1+p-1},\ldots,\inputv{s}^{\#s-1},
      \outputv{s}^{\#s-1}),\\
      & b_n(\outputv{\vars},\outputv{s}_b)
    \end{align*}
    % for each $1\leq i \leq n$.
  \item If $n = 0$, then $\#s=1$ and
    $b_{\mathit{CLP}}$ is the set consisting 
    of the $\clppl$ clause
    \[b(\inputv{\vars},\inputv{s}_b) \la 
    \{\inputv{s}^0\geq 1,\ \inputv{s}_b=\outputv{s}^0\},
    b_m(\inputv{s}^{p-1},\ldots,\inputv{s}^0, \outputv{s}^0)\]
  \end{enumerate}
  \qed
\end{definition}

The return type of the methods that we consider in
Definitions~\ref{def:clp_block_no_call_not_void}--%
\ref{def:clp_block_call_not_void_one_inst} is supposed to be
non-\verb+void+.
The situation where block $b$ occurs inside a method whose
return type is \texttt{void} is handled similarly,
except that we remove the variables $\inputv{s}_b$ and $\outputv{s}_b$
and the constraints where they occur.
The situation where the first instruction of $b$ is
$\mathsf{call}\ m$, where the return type
of $m$ is $\mathtt{void}$,
is handled as in Definitions~\ref{def:clp_block_call_not_void}%
--\ref{def:clp_block_call_not_void_one_inst}, except that we
remove the last parameter of $b_m$.
In Definitions~\ref{def:clp_block_call_not_void}--%
\ref{def:clp_block_call_not_void_one_inst}, $m$ is also
supposed to be a non-static method. The situation where $m$
is static is handled similarly, except that we remove 
the constraint $\inputv{s}^{\#s-1}\geq 1$, as there is no
call receiver on the stack.
\begin{definition}\label{def:clp_construction}
  Let $P$ be a Java bytecode program given as a graph of blocks.
  The $\clppl$ program $P_{\mathit{CLP}}$ derived from $P$
  is defined as
  \[P_{\mathit{CLP}} = \bigcup_{b\in P}b_{\mathit{CLP}}\;.\]
  \qed
\end{definition}

In this paper, we consider the \emph{leftmost selection rule}
for computations in $\clppl$ programs. Then, we have:
%
%%% Theorem: non-term CLP => non-term JBC
\begin{theorem}\label{theorem:non-term}
  Let $J$ be a Java Virtual Machine,
  $P$ be a Java bytecode program consisting of
  instructions in $\correctins$, and
  $b$ be a block of $P$.
  Let $\mathit{vars}$ be some fixed integer values and
  $\inputv{s}_b$ be a free variable.
  If the query $b(\mathit{vars}, \inputv{s}_b)$ has an infinite
  computation in $P_{\mathit{CLP}}$ then there is an execution of $J$
  started at block $b$ that does not terminate.
\end{theorem}

\begin{figure}[t]
  \lstset{language=Java, numbers=left, numberstyle=\tiny,
    basicstyle=\tt\scriptsize, frame=shadowbox, rulesepcolor=\color{blue}}
  
  \begin{lstlisting}
    public class Sum {
      public static int sum(int n) {
        if (n == 0) return 0;
        else return n + sum(n - 1);
      }
      public static void main(String args[]) {
        sum(-1);
      }
    }
  \end{lstlisting}
  \caption{A program with a recursive method \texttt{sum} that
    takes an integer $n$ as input and computes the sum $0+1+\cdots+n$.
    This method does not terminate on negative inputs.}
  \label{fig:sum_term_java}
\end{figure}

\begin{example}\label{example:sumNonTerm}
  Consider the recursive method \texttt{sum} in
  Fig.~\ref{fig:sum_term_java}, whose graph of blocks
  is given in Fig.~\ref{fig:sum_nonterm_blocks}.
  \begin{figure}[t]
    \begin{center}
      \includegraphics[scale=0.41]{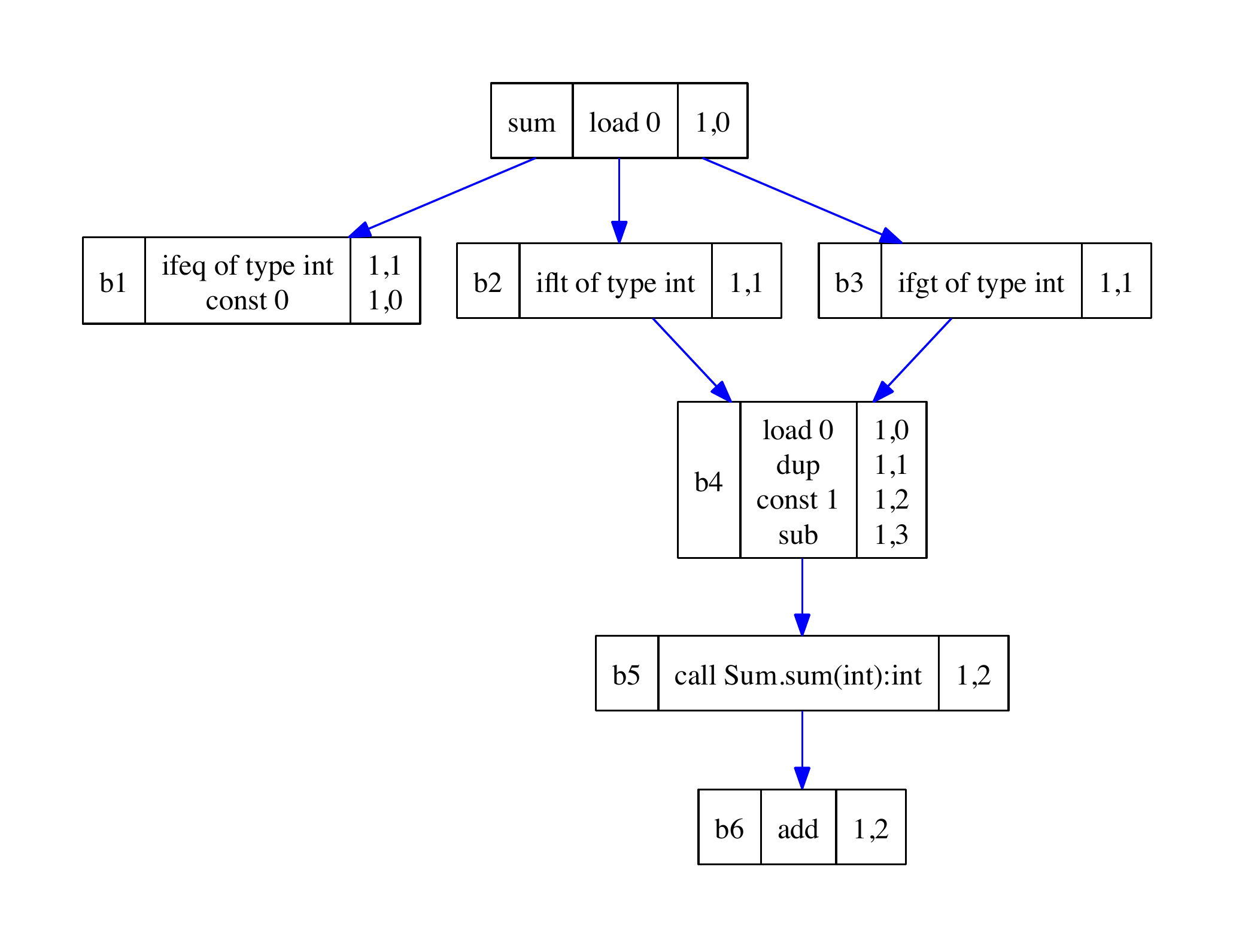}
    \end{center}
    \caption{The graph of blocks of the method \texttt{sum}
      in Fig.~\ref{fig:sum_term_java}, where each block
      is decorated with a unique name. On the right of each instruction,
      we report the number of local variables and stack elements at
      that program point, just before executing the instruction.}
    \label{fig:sum_nonterm_blocks}
  \end{figure}
  \begin{itemize}
  \item The block $\mathit{sum}$ has $b_1$, $b_2$ and $b_3$ as successors
    and its first instruction is not a $\mathsf{call}$ instruction.
    Let $q$ be the program point where the instruction
    $\mathsf{load}\ 0$ of the block $\mathit{sum}$ occurs.
    We have
    \begin{align*}
      \mathit{load}_q^{\mathbb{PL}}\ 0
      &= \mathit{Unchanged}_q(1,0)\cup\{\inputv{l}^0=\outputv{s}^0\}\\
      &= \{\inputv{l}^0=\outputv{l}^0,\inputv{l}^0=\outputv{s}^0\}\;.
    \end{align*}
    Hence, by Definition~\ref{def:clp_block_no_call_not_void},
    $\mathit{sum}_{\mathit{CLP}}$ consists of the following clauses:
    \begin{align*}
      \mathit{sum}(\inputv{l}^0,\inputv{s}_{\mathit{sum}}) \la \; &
      \{\inputv{l}^0=\outputv{l}^0,\inputv{l}^0=\outputv{s}^0,
      \inputv{s}_{\mathit{sum}}=\outputv{s}_{\mathit{sum}}\},\\
      & b_1(\outputv{l}^0,\outputv{s}^0,\outputv{s}_{\mathit{sum}})
      \\
      \mathit{sum}(\inputv{l}^0,\inputv{s}_{\mathit{sum}}) \la \; &
      \{\inputv{l}^0=\outputv{l}^0,\inputv{l}^0=\outputv{s}^0,
      \inputv{s}_{\mathit{sum}}=\outputv{s}_{\mathit{sum}}\},\\
      & b_2(\outputv{l}^0,\outputv{s}^0,\outputv{s}_{\mathit{sum}})
      \\
      \mathit{sum}(\inputv{l}^0,\inputv{s}_{\mathit{sum}}) \la \; &
      \{\inputv{l}^0=\outputv{l}^0,\inputv{l}^0=\outputv{s}^0,
      \inputv{s}_{\mathit{sum}}=\outputv{s}_{\mathit{sum}}\},\\
      & b_3(\outputv{l}^0,\outputv{s}^0,\outputv{s}_{\mathit{sum}})\;.
    \end{align*}
  \item The block $b_1$ has no successor
    and its first instruction is not a $\mathsf{call}$ instruction.
    Let $q$ and $q'$ be the program points where the instructions
    $\mathsf{ifeq\ of\ type}\ \mathtt{int}$ and $\mathsf{const}\ 0$
    of the block $b_1$ occur, respectively.
    We have
    \begin{align*}
      & \;(\mathit{ifeq\ of\ type}_q^{\mathbb{PL}}\ \mathtt{int})
      ;^{\mathbb{PL}} (\mathit{const}_{q'}^{\mathbb{PL}}\ 0)\\
      = & \;(\mathit{Unchanged}_q(1,0)\cup\{\inputv{s}^0 = 0\})\\
      & \quad;^{\mathbb{PL}} (\mathit{Unchanged}_{q'}(1,0)\cup\{0=\outputv{s}^0\}) \\
      = & \;\{\inputv{l}^0=\outputv{l}^0,\inputv{s}^0 = 0\}
      ;^{\mathbb{PL}}\{\inputv{l}^0=\outputv{l}^0,0=\outputv{s}^0\}\\
      = & \;\{\inputv{l}^0=\outputv{l}^0,\inputv{s}^0 = 0,0=\outputv{s}^0\} \;.
    \end{align*}
    Hence, by Definition~\ref{def:clp_block_no_call_not_void},
    ${b_1}_{\mathit{CLP}}$ consists of the following clause:
    \[\begin{array}{l}
      b_1(\inputv{l}^0,\inputv{s}^0,\inputv{s}_{b_1}) \la 
      \{\inputv{l}^0=\outputv{l}^0,\inputv{s}^0 = 0,0=\outputv{s}^0,
      \inputv{s}_{b_1}=\outputv{s}^0\}\;.
    \end{array}\]
  \item The block $b_5$ consists of only one instruction,
    which is a call to a static method,
    and it has $b_6$ as unique successor. We have
    \[c_==\{\inputv{l}^0=\outputv{l}^0,\inputv{s}^0=\outputv{s}^0\}\;.\]
    Hence, by Definition~\ref{def:clp_block_call_not_void_one_inst}, 
    ${b_5}_{\mathit{CLP}}$ consists of the following clause:
    \begin{align*}
      b_5(\inputv{l}^0,\inputv{s}^0,\inputv{s}^1,\inputv{s}_{b_5}) \la \; &
      \{\inputv{l}^0=\outputv{l}^0,\inputv{s}^0=\outputv{s}^0,
      \inputv{s}_{b_5}=\outputv{s}_{b_5}\},\\
      & \mathit{sum}(\inputv{s}^1,\outputv{s}^1),
      b_6(\outputv{l}^0,\outputv{s}^0,\outputv{s}^1,\outputv{s}_{b_5})\;.
    \end{align*}
  \end{itemize}

  The complete $\clppl$ program derived from the program
  in Fig.~\ref{fig:sum_term_java}
  consists of the clauses in Table~\ref{table:clp}.
  \qed
  \begin{table}[t]
    \small
    \[\begin{array}{|cl|}
      % 
      %%% sum
      %
      \hline
      \multicolumn{2}{|c|}{\texttt{sum}}\\
      \hline
      (r_1) & \mathit{sum}(\inputv{l}^0,\inputv{s}_{\mathit{sum}}) \la
      \{\inputv{l}^0=\outputv{l}^0,\inputv{l}^0=\outputv{s}^0,
      \inputv{s}_{\mathit{sum}}=\outputv{s}_{\mathit{sum}}\},\\
      & \hspace{2.1cm}b_1(\outputv{l}^0,\outputv{s}^0,\outputv{s}_{\mathit{sum}})
      \\[1ex]
      (r_2) & \mathit{sum}(\inputv{l}^0,\inputv{s}_{\mathit{sum}}) \la
      \{\inputv{l}^0=\outputv{l}^0,\inputv{l}^0=\outputv{s}^0,
      \inputv{s}_{\mathit{sum}}=\outputv{s}_{\mathit{sum}}\},\\
      & \hspace{2.1cm}b_2(\outputv{l}^0,\outputv{s}^0,\outputv{s}_{\mathit{sum}})
      \\[1ex]
      (r_3) & \mathit{sum}(\inputv{l}^0,\inputv{s}_{\mathit{sum}}) \la
      \{\inputv{l}^0=\outputv{l}^0,\inputv{l}^0=\outputv{s}^0,
      \inputv{s}_{\mathit{sum}}=\outputv{s}_{\mathit{sum}}\},\\
      & \hspace{2.1cm}b_3(\outputv{l}^0,\outputv{s}^0,\outputv{s}_{\mathit{sum}})
      \\[1ex]
      (r_4) & b_1(\inputv{l}^0,\inputv{s}^0,\inputv{s}_{b_1}) \la
      \{\inputv{l}^0=\outputv{l}^0,\inputv{s}^0 = 0,0=\outputv{s}^0,
      \inputv{s}_{b_1}=\outputv{s}^0\}
      \\[1ex]
      (r_5) & b_2(\inputv{l}^0,\inputv{s}^0,\inputv{s}_{b_2}) \la
      \{\inputv{l}^0=\outputv{l}^0,\inputv{s}^0 \leq -1,\inputv{s}_{b_2}=\outputv{s}_{b_2}\},\\
      & \hspace{2cm}b_4(\outputv{l}^0,\outputv{s}_{b_2})
      \\[1ex]
      (r_6) & b_3(\inputv{l}^0,\inputv{s}^0,\inputv{s}_{b_3}) \la
      \{\inputv{l}^0=\outputv{l}^0,\inputv{s}^0 \geq 1,\inputv{s}_{b_3}=\outputv{s}_{b_3}\},\\
      & \hspace{2cm}b_4(\outputv{l}^0,\outputv{s}_{b_3})
      \\[1ex]
      (r_7) & b_4(\inputv{l}^0,\inputv{s}_{b_4}) \la
      \{\inputv{l}^0=\outputv{l}^0,
      \inputv{l}^0=\outputv{s}^0,\inputv{l}^0-1=\outputv{s}^1,
      \inputv{s}_{b_4}=\outputv{s}_{b_4}\},\\
      & \hspace{1.6cm}
      b_5(\outputv{l}^0,\outputv{s}^0,\outputv{s}^1,\outputv{s}_{b_4})
      \\[1ex]
      (r_8) & b_5(\inputv{l}^0,\inputv{s}^0,\inputv{s}^1,\inputv{s}_{b_5}) \la
      \{\inputv{l}^0=\outputv{l}^0,\inputv{s}^0=\outputv{s}^0,
      \inputv{s}_{b_5}=\outputv{s}_{b_5}\},\\
      & \hspace{2.5cm}\mathit{sum}(\inputv{s}^1,\outputv{s}^1),
      b_6(\outputv{l}^0,\outputv{s}^0,\outputv{s}^1,\outputv{s}_{b_5})
      \\[1ex]
      (r_9) & b_6(\inputv{l}^0,\inputv{s}^0,\inputv{s}^1,\inputv{s}_{b_6}) \la
      \{\inputv{l}^0=\outputv{l}^0,\inputv{s}^0+\inputv{s}^1=\outputv{s}^0,
      \inputv{s}_{b_6}=\outputv{s}^0\}
      \\
      % 
      %%% main
      % 
      \hline\hline
      \multicolumn{2}{|c|}{\texttt{main}}\\
      \hline
      (r_{10}) & \mathit{main}(\inputv{l}^0) \la
      \{\inputv{l}^0=\outputv{l}^0,-1=\outputv{s}^0\},
      b_7(\outputv{l}^0,\outputv{s}^0)
      \\[1ex]
      (r_{11}) & b_7(\inputv{l}^0,\inputv{s}^0) \la
      \{\inputv{l}^0=\outputv{l}^0\},
      \mathit{sum}(\inputv{s}^0,\outputv{s}^0),
      b_8(\outputv{l}^0,\outputv{s}^0)
      \\[1ex]
      (r_{12}) & b_8(\inputv{l}^0,\inputv{s}^0) \la
      \{\inputv{l}^0=\outputv{l}^0\}\\
      \hline
    \end{array}\]
  \caption{The $\clppl$ program derived from the program in
    Fig.~\ref{fig:sum_term_java}}\label{table:clp}
  \end{table}
\end{example}

%%%%%%%%%%%%%%%%%%%%%%%%%%%%%%%%%%%%%%%%%%%%%%%%%%%%%%%%%%%%%%%%%

%====================================
\section{Proving Non-Termination}
\label{section:clp-nontermination}
%====================================
%
Let $P$ be a Java bytecode program and $b$ be a block
of $P$.
By Theorem~\ref{theorem:non-term}, the existence of
an infinite computation in $P_{\mathit{CLP}}$ of a query
of the form $b(\cdots)$ entails that there
is an execution of the Java Virtual Machine starting at
$b$ that does not terminate. 
In~\cite{Payet09a}, we provide a criterion that can
be used for proving the existence of a non-terminating
query for a \emph{binary} $\clppl$ program \ie a
program consisting of clauses whose body contains at
most one atom (we refer to such clauses as binary
clauses).
Note that by Definitions~\ref{def:clp_block_call_not_void}--%
\ref{def:clp_block_call_not_void_one_inst},
$P_{\mathit{CLP}}$ may not be binary as it may contain clauses
whose body consists of two atoms.

We use the binary unfolding operation~\cite{Gabbrielli94a}
to transform $P_{\mathit{CLP}}$ into a binary program.
We write $\clppl$ clauses as
\[H \la c, B_1, \ldots , B_m \quad\text{or}\quad
H \la c,\mathbf{B}\]
where $c$ is a $\clppl$ constraint,
$H$, $B_1$, \ldots, $B_m$ are atoms and
$\mathbf{B}$ is a sequence of atoms.
When $B_1, \ldots , B_m$ or $\mathbf{B}$ are empty
sequences, we write the clause as $H \la c$.
We let $\mathit{id}$ denote the set of all the binary clauses
of the form
\[p(\tilde{x}) \la \{\tilde{x}=\tilde{y}\}, p(\tilde{y})\]
where $p$ is a predicate symbol and $\tilde{x}$ and $\tilde{y}$
are disjoint sequences of distinct variables. 
We also let $\overline{\exists}_{\mathit{Var}(H,\mathbf{B})}[\cdots]$
denote the projection of $[\cdots]$ onto the variables of
$H$ and $\mathbf{B}$.
%
%%% Definition: binunf
\begin{definition}\label{definition:binunf}
  Let $P_{\mathit{CLP}}$ be a $\clppl$ program and $X$
  be a set of binary clauses. Then,
  \[T_{P_{\mathit{CLP}}}^{\beta}(X) = \mathcal{F} \cup \mathcal{B}\]
  where 
  \[\mathcal{F}
  = \big\{H \la c \in P_{\mathit{CLP}} \, | \,
  c \text{ is satisfiable}\big\}\]
  and $\mathcal{B}$ is the set
  {\small
    \[\begin{array}{l}
      \left\{
        H \la c, \mathbf{B}
        \begin{array}{@{\;}|@{\;}l}
          r= H \la c_0,
          B_1, \ldots , B_m \in P_{\mathit{CLP}},\
          1 \leq i \leq m\\[1ex]
          \langle H_j \leftarrow c_j\rangle_{j=1}^{i-1}
          \in X
          \mbox{ renamed apart from $r$}\\[1ex]
          H_i \leftarrow c_i, \mathbf{B} \in X \cup \mathit{id}
          \mbox{ renamed apart from $r$}\\[1ex]
          i<m \Rightarrow \mathbf{B} \neq \emptyset\\[1ex]
          c  = \overline{\exists}_{\mathit{Var}(H,\mathbf{B})} \big[
          c_0 \land \mathop{\land}\limits_{j=1}^i (c_j \land
          B_j=H_j)\big]\\[1ex]
          c \mbox{ is satisfiable}
        \end{array}\right\}
    \end{array}\]}
  
  We define the powers of $T_{P_{\mathit{CLP}}}^{\beta}$ as usual:
  \begin{align*}
    T_{P_{\mathit{CLP}}}^{\beta} \uparrow 0 & = \emptyset\\
    T_{P_{\mathit{CLP}}}^{\beta} \uparrow (i + 1) & =
    T_{P_{\mathit{CLP}}}^{\beta}(T_{P_{\mathit{CLP}}}^{\beta}  \uparrow i)
    \quad\forall i \in \nat\;.
  \end{align*}
  It can be shown that the least fixpoint ($\mathit{lfp}$) 
  of this monotonic operator always exists and we set
  \[\mathit{binunf}(P_{\mathit{CLP}})=\mathit{lfp}
  (T_{P_{\mathit{CLP}}}^{\beta})\;.\]
  \qed
\end{definition}

\begin{example}\label{example:proving_non_termination}
  Consider the program $P_{\mathit{CLP}}$
  in Table~\ref{table:clp}.
  
  We have $T_{P_{\mathit{CLP}}}^{\beta} \uparrow 0 = \emptyset$. Moreover,
  \begin{itemize}
  \item the set $T_{P_{\mathit{CLP}}}^{\beta} \uparrow 1$ includes the
    following clauses:
    \[\begin{array}{cr@{\;\la\;}l}
      (u_1)
      & b_5(\inputv{l}^0,\inputv{s}^0,\inputv{s}^1,\inputv{s}_{b_5})
      &  \{\}, \mathit{sum}(\inputv{s}^1,\outputv{s}^1)
      \\[1ex]
      (u_2)
      & b_7(\inputv{l}^0,\inputv{s}^0)
      & \{\},\mathit{sum}(\inputv{s}^0,\outputv{s}^0)
    \end{array}\]
    where $u_1$ is obtained by unfolding $r_8$ with $\mathit{id}$ and
    $u_2$ by unfolding $r_{11}$ with $\mathit{id}$,
  \item the set $T_{P_{\mathit{CLP}}}^{\beta} \uparrow 2$ includes the
    following clauses:
    \[\begin{array}{cr@{\;\la\;}l}
      (u_3)
      & b_4(\inputv{l}^0,\inputv{s}_{b_4}) 
      & \{\inputv{l}^0 - 1 = \inputv{s}^1\},
      \mathit{sum}(\inputv{s}^1,\outputv{s}^1)
      \\[1ex]
      (u_4)
      & \mathit{main}(\inputv{l}^0)
      & \{-1=\inputv{s}^0\},
      \mathit{sum}(\inputv{s}^0,\outputv{s}^0)
    \end{array}\]
    where $u_3$ is obtained by unfolding $r_7$ with $u_1$ and
    $u_4$ by unfolding $r_{10}$ with $u_2$,
  \item the set $T_{P_{\mathit{CLP}}}^{\beta} \uparrow 3$ includes the
    following clause, obtained by unfolding $r_5$ with $u_3$:
    \begin{align*}
      (u_5) \quad
      b_2(\inputv{l}^0,\inputv{s}^0,\inputv{s}_{b_2}) \la \; &
      \{\inputv{l}^0-1=\inputv{s}^1,
      \inputv{s}^0 \leq -1\},\\
      & \mathit{sum}(\inputv{s}^1,\outputv{s}^1)
    \end{align*}
  \item the set $T_{P_{\mathit{CLP}}}^{\beta} \uparrow 4$ includes the
    following clause, obtained by unfolding $r_2$ with $u_5$:
    \begin{align*}
      (u_6)
      \quad \mathit{sum}(\inputv{l}^0,\inputv{s}_{\mathit{sum}}) \la \; &
      \{\inputv{l}^0 \leq -1, \inputv{l}^0-1=\inputv{s}^1\},\\
      & \mathit{sum}(\inputv{s}^1,\outputv{s}^1)\;.
    \end{align*}
    \qed
  \end{itemize}
\end{example}

It is proved in~\cite{Codish99a} that existential
non-termination in $P_{\mathit{CLP}}$ is \emph{equivalent}
to existential non-termination in $\mathit{binunf}(P_{\mathit{CLP}})$:
%
%%% Theorem: Observing Termination
\begin{theorem}
  \label{theorem-observing-termination}
  Let $P_{\mathit{CLP}}$ be a $\clppl$ program
  and $Q$ be a query consisting of one atom. Then, $Q$ has an
  infinite computation in $P_{\mathit{CLP}}$ 
  \emph{if and only if} $Q$ has an infinite computation in
  $\mathit{binunf}(P_{\mathit{CLP}})$.
\end{theorem}

Note that $\mathit{binunf}(P_{\mathit{CLP}})$ is a possibly
infinite set of binary clauses. In practice, we compute
only the first $\mathit{max}$ iterations of
$T_{P_{\mathit{CLP}}}^{\beta}$, where $\mathit{max}$ is a parameter
of the analysis, and we have
$T_{P_{\mathit{CLP}}}^{\beta} \uparrow \mathit{max}
\subseteq \mathit{binunf}(P_{\mathit{CLP}})$. Therefore,
any query that has an infinite computation in
$T_{P_{\mathit{CLP}}}^{\beta} \uparrow \mathit{max}$
also has an infinite computation in
$\mathit{binunf}(P_{\mathit{CLP}})$,
hence, by Theorem~\ref{theorem-observing-termination},
in $P_{\mathit{CLP}}$.

In the results we present below, $p$ is a predicate symbol,
$\tilde{x}$ and $\tilde{y}$ are disjoint sequences of
distinct variables and $c$ is a $\clppl$ constraint
on $\tilde{x}$ and $\tilde{y}$ only (\ie the set of
variables appearing in $c$ is a subset of 
$\tilde{x} \cup \tilde{y}$).
The criterion that we provide in~\cite{Payet09a}
for proving the existence of a non-terminating
query can be formulated as follows in the context
of $\clppl$ clauses.
%
%%% Proposition: TPLP criterion
\begin{proposition}
  \label{proposition:unit-loop-tplp}
  Let
  \[r = p(\tilde{x}) \la c, p(\tilde{y})\]
  be a $\clppl$ binary clause where $c$ is satisfiable.
  Let $e(\tilde{x})$ denote the projection of
  $c$ onto $\tilde{x}$.
  Suppose the formula
  \[\forall \tilde{x}\ e(\tilde{x}) \Rightarrow
  \big[\forall \tilde{y}\ c \Rightarrow 
  e(\tilde{y})\big]\]
  is true. Then, $p(\tilde{v})$ has an infinite
  computation in $\{r\}$, for some fixed integer values
  for $\tilde{v}$.
\end{proposition}
The sense of Proposition~\ref{proposition:unit-loop-tplp}
is the following. Satisfiability for $c$ means that there exists
some ``input'' to the clause \ie some value from which one
can ``enter'' the clause. The logical formula means:
let $\tilde{a}$ be some input to the clause (\ie 
$\forall \tilde{x}\ e(\tilde{x})$), then any output
$\tilde{b}$ corresponding to $\tilde{a}$ (\ie
$\forall \tilde{y}\ c$) is also an input to the
clause (\ie $e(\tilde{y})$). Shortly, if one can enter the
clause with $\tilde{a}$, then one can enter the clause again
with \emph{any} output corresponding to $\tilde{a}$.
This corresponds to a notion of \emph{unavoidable}
(\emph{universal}) non-termination, as any input
to the clause necessarily leads to an infinite
computation.

The criterion provided in
Proposition~\ref{proposition:unit-loop-tplp}
can be refined into:
%
%%% Proposition: improvement of TPLP criterion
\begin{proposition}
  \label{proposition:unit-loop}
  Let
  \[r = p(\tilde{x}) \la c, p(\tilde{y})\]
  be a $\clppl$ binary clause where $c$ is satisfiable.
  Let $e(\tilde{x})$ denote the projection of
  $c$ onto $\tilde{x}$.
  Suppose the formula
  \[\forall \tilde{x}\ e(\tilde{x}) \Rightarrow
  \big[\exists \tilde{y}\ c \land
  e(\tilde{y})\big]\]
  is true. Then, $p(\tilde{v})$ has an infinite
  computation in $\{r\}$, for some fixed integer values
  for $\tilde{v}$.
\end{proposition}
Now, the sense of the logical formula is:
if one can enter the clause with an input $\tilde{a}$, then
\emph{there exists} an output $\tilde{b}$ corresponding
to $\tilde{a}$ such that one can enter the clause
again with $\tilde{b}$.
This corresponds to a notion of \emph{potential}
(\emph{existential}) non-termination, as for any input
to the clause there is a corresponding output that leads
to an infinite computation.

The criterion of Proposition~\ref{proposition:unit-loop-tplp}
entails that of Proposition~\ref{proposition:unit-loop}, but
the converse does not hold in general (\eg for the clause
$p(x) \la \{x \geq 0\}, p(y)$ the logical formula
of Proposition~\ref{proposition:unit-loop} is true
whereas that of Proposition~\ref{proposition:unit-loop-tplp}
is not).

In Proposition~\ref{proposition:unit-loop}, we consider
\emph{recursive} binary clauses. A recursive clause
in $\mathit{binunf}(P_{\mathit{CLP}})$ is a \emph{compressed}
form of a loop in $P_{\mathit{CLP}}$. The next result 
allows one to ensure that a loop is reachable
from a given program point.
%
%%% Proposition: accessible loops
\begin{proposition}
  \label{proposition:compound-loops}
  Let
  \[\begin{array}{r@{\;=\;}l}
    r & p(\tilde{x}) \la c, p(\tilde{y})\\[1.5ex]
    r' & p'(\tilde{x}') \la c', p(\tilde{y}')
  \end{array}\]
  be some $\clppl$ binary clauses where $c$ and
  $c'$ are satisfiable.
  Let $e(\tilde{x})$ denote the projection of
  $c$ onto $\tilde{x}$.
  Suppose the formul\ae
  \begin{itemize}
  \item $\forall \tilde{x}\ e(\tilde{x}) \Rightarrow
    \big[\exists \tilde{y}\ c \land e(\tilde{y})\big]$
  \item $\exists \tilde{x}',\ \exists \tilde{y}',\
    c' \land e(\tilde{y}')$
  \end{itemize}
  are true. Then, $p'(\tilde{v})$ has an infinite
  computation in $\{r,r'\}$, for some fixed integer values
  for $\tilde{v}$.
\end{proposition}
The sense of the second logical formula
in Proposition~\ref{proposition:compound-loops}
is that there is an output to the clause
$r'$ which is an input to the clause $r$.
Moreover, any input to $r'$ that satisfies the
second formula is the starting point of a
potential infinite computation; if the first
logical formula in Proposition~\ref{proposition:compound-loops}
was that of Proposition~\ref{proposition:unit-loop-tplp},
then it would be the starting point of an
unavoidable infinite computation.

\begin{example}
  The clauses $u_4$ and $u_6$ in
  Example~\ref{example:proving_non_termination}
  satisfy the pre-conditions
  of Proposition~\ref{proposition:compound-loops}.
  Hence, by Proposition~\ref{proposition:compound-loops},
  there is a value $v$ in $\mathbb{Z}$ which is such that
  $\mathit{main}(v)$ has an infinite computation in
  $\mathit{binunf}(P_{\mathit{CLP}})$, hence an infinite
  computation in $P_{\mathit{CLP}}$. Consequently,
  by Theorem~\ref{theorem:non-term}, there is an execution
  of the Java virtual Machine started at block $\mathit{main}$
  that does not terminate, where $\mathit{main}$ is the
  initial block of the Java bytecode program corresponding
  to the Java program in Fig.~\ref{fig:sum_term_java}.
  \qed
\end{example}

%%%%%%%%%%%%%%%%%%%%%%%%%%%%%%%%%%%%%%%%%%%%%%%%%%%%%%%%%%%%%%%%%

%=======================
\section{Experiments}
\label{section:experiments}
%=======================
We implemented our approach in the Julia analyser.
Non-termination proofs of $\clppl$ programs
are performed by the BinTerm tool, a component of Julia
that implements Proposition~\ref{proposition:compound-loops}.
BinTerm is written in SWI-Prolog~\cite{wielemaker11}
and relies on the Parma Polyhedra Library~\cite{BagnaraHZ08SCP}
for checking satisfiability of integer linear constraints. 
Elimination of existentially quantified variables in $\integers$ 
follows the approach of the  Omega Test~\cite{Pugh92}.

We evaluated our analyser on a collection of 113
examples, made up of:
\begin{itemize}
\item a set of 75 iterative examples, consisting of 54 programs
  provided by~\cite{InvelWeb,VelroyenR08}%  
  \footnote{We removed the incomplete program \texttt{factorial}
    from the collection
    of simple examples from~\cite{InvelWeb,VelroyenR08}.}
  and the 21 non-terminating
  programs submitted by the Julia team to the 
  \emph{International Termination Competition} 
  \cite{termcomp} in 2011,
\item a set of 38 recursive examples, consisting of 34 programs that
  we obtained by turning some examples from~\cite{InvelWeb,VelroyenR08}
  into recursive programs, and 4 programs that we wrote ourselves;
  all of these recursive programs do not terminate due to an infinite
  recursion.
\end{itemize}

In our experiments, we compared AProVE, Invel and the new version
of Julia. We used the default settings for each tool \ie a
time-out of 60 seconds for AProVE, a maximum number of iterations
set to 10 for Invel and, for Julia, a time-out of 20 seconds for
each strongly connected component of the $\clppl$ program.
% and a maximum number of iterations of the binary unfolding operator
%$T_{P_{\mathit{CLP}}}^{\beta}$ set to 13.
Details on the experiments are available at~\cite{PayetMP12experiments}.
We do not provide running times as we could not run all the
tools on a same machine (the AProVE implementation that performs
non-termination proofs is available through a web interface only).

Table~\ref{table:results-imp} and Table~\ref{table:results-rec}
give an overview of the results that we obtained. 
Here, ``Invel'' is the set of 54 examples from~\cite{InvelWeb,VelroyenR08},
``Julia TC11'' is the set of 21 non-terminating examples
submitted by Julia to the competition in 2011,
``Invel rec.'' is the set obtained by turning 34 Invel
examples into recursive programs and
 ``Julia rec.'' corresponds to the 4 programs that we wrote.
Moreover, \textbf{Y} and \textbf{N} indicate how often termination
(resp. non-termination) could be proved, \textbf{F} indicates
how often a tool failed within the time limit set by its default
settings and \textbf{T} states how many examples led to time-out.
Finally, \textbf{P} (\textbf{P}roblems) gives the number of times
a tool stopped with a run-time exception or produced an incorrect
answer: the Invel analyser issued 2 incorrect answers (2 terminating
programs incorrectly proved non-terminating) and stopped twice
with a \texttt{NullPointerException}.
Table~\ref{table:results-imp} shows that Julia failed on 22 Invel
examples; 13 of these failures are due to the use, in the corresponding
programs, of bytecode instructions that are not exactly abstracted 
into constraints.
Table~\ref{table:results-rec} clearly shows that only Julia could
detect non-terminating recursions.

\begin{table}[t]
  \begin{center}
    \begin{tabular}{|l|c|c|c|c|c||c|c|c|c|}
      \cline{2-10}
      \multicolumn{1}{c|}{}
      & \multicolumn{5}{c||}{Invel (54)}
      & \multicolumn{4}{c|}{Julia TC11 (21)}\\
      \cline{2-10}
      \multicolumn{1}{c|}{}
      & \textbf{Y} & \textbf{N} & \textbf{F} & \textbf{T} & \textbf{P}
      & \textbf{N} & \textbf{F} & \textbf{T} & \textbf{P} \\
      \hline
      AProVE
      & 3 & 49 & 0 & 2 & 0
      & 21 & 0 & 0 & 0 \\
      \hline
      Invel
      & 0 & 29 & 1 & 21 & 3
      & 13 & 0 & 7 & 1\\
      \hline
      Julia
      & 2 & 30 & 22 & 0 & 0
      & 17 & 4 & 0 & 0 \\
      \hline
    \end{tabular}
  \end{center}
  \caption{The results of our experiments on iterative programs.}
  \label{table:results-imp}
\end{table}

\begin{table}[t]
  \begin{center}
    \begin{tabular}{|l|c|c|c||c|c|}
      \cline{2-6}
      \multicolumn{1}{c|}{}
      & \multicolumn{3}{c||}{Invel rec. (34)}
      & \multicolumn{2}{c|}{Julia rec. (4)}\\
      \cline{2-6}
      \multicolumn{1}{c|}{}
      & \textbf{N} & \textbf{F} & \textbf{T}  
      & \textbf{N} & \textbf{T} \\
      \hline
      AProVE
      & 0 & 0 & 34 
      & 0 & 4\\
      \hline
      Invel
      & 0 & 0 & 34
      & 0 & 4 \\
      \hline
      Julia
      & 26 & 8 & 0
      & 4  & 0 \\
      \hline
    \end{tabular}
  \end{center}
  \caption{The results of our experiments on recursive programs.}
  \label{table:results-rec}
\end{table}

\subsection*{Interpreting the Web Interface of Julia}
%===================================================
For our experiments, there are three cases.
Let us consider the Invel examples.

\begin{itemize}
\item For \texttt{alternatingIncr.jar},
we get one warning:
\begin{verbatim}
[Termination] are you sure that 
simple.alternatingIncr.increase always
terminates?
\end{verbatim}
It means that Julia has a proof that the upper approximation of the program
built to prove termination loops in $\mathbb{Q}$,
indicating a potential non-termination of the original program.

\item For \texttt{alternDiv.jar}, we get two warnings:
\begin{verbatim}
[Termination] are you sure that 
simple.alternDiv.AlternDiv.loop always
terminates?

[Termination] simple.alternDiv.Main.main may 
actually diverge
\end{verbatim}
The first warning has the same interpretation as above
while the second one emphasizes that Julia has a
non-termination proof for the original program.

\item For \texttt{whileBreak.jar}, there are no warnings.
It means that Julia has a termination proof for the original program.
\end{itemize}

%%%%%%%%%%%%%%%%%%%%%%%%%%%%%%%%%%%%%%%%%%%%%%%%%%%%%%%%%%%%%%%%%

%======================
\section{Conclusion}
\label{section:conclusion}
%======================
The work we have presented in this paper was initiated
in~\cite{Payet09b}, where we observed that in some situations
the non-termination of a Java bytecode program can be deduced
from that of its $\clppl$ translation.
Here, we have introduced the formal material that is missing
in~\cite{Payet09b} and we have presented a new experimental
evaluation, conducted with the Julia tool which now includes
an implementation of our results.
Currently, our non-termination analysis cannot be applied to programs
that include certain types of object field access.
We are actively working at extending its scope, so that it can
identify sources of non-termination such as traversal of cyclical
data structures.

%%%%%%%%%%%%%%%%%%%%%%%%%%%%%%%%%%%%%%%%%%%%%%%%%%%%%%%%%%%%%%%%%
%\section*{Acknowledgements}

%\newpage
%%%%%%%%%%%%%%%%%%%%%%%%%%%%%%%%%%%%%%%%%%%%%%%%%%%%%%%%%%%%%%%%%
\bibliographystyle{abbrv}
% \bibliography{ppdp}

%%%%%%%%%%%%%%%%%%%%%%%%%%%%%%%%%%%%%%%%%%%%%%%%%%%%%%%%%%%%%%%%%
\newpage
\appendix

\section{Proofs}

%%%%%%%%
\subsection{Proposition~\ref{proposition:composition}}
%%%%%%%%
Let $\delta_1\in\Delta_{l_i,s_i\to l_t,s_t}$,
$\delta_2\in\Delta_{l_t,s_t\to l_o,s_o}$,
$\pl_1\in\mathbb{PL}_{l_i,s_i\to l_t,s_t}$ and
$\pl_2\in\mathbb{PL}_{l_t,s_t\to l_o,s_o}$.
Suppose that $\pl_1\models\delta_1$, that $\pl_2\models\delta_2$
and that $\delta_1$ is compatible with $\delta_2$.
For sake of readability, we let $\pl_{1,2}$ denote
the constraint $\pl_1;^{\mathbb{PL}}\pl_2$.
We have $\delta_1;\delta_2\in\Delta_{l_i,s_i\to l_o,s_o}$ and
$\pl_{1,2}\in\mathbb{PL}_{l_i,s_i\to l_o,s_o}$.
We have to prove that
$\pl_{1,2}\models\delta_1;\delta_2$. So, consider
any model $\rho$ of $\pl_{1,2}$ and
any state $\sigma$ compatible with $\delta_1;\delta_2$.
Suppose that 
$\inputv{len}(\sigma) = \inputv{\rho}$. We have to prove that
$(\delta_1;\delta_2)(\sigma)$ is defined and
$\outputv{len}((\delta_1;\delta_2)(\sigma)) =
\outputv{\rho}$.

Let
\[T=\{\overline{l}^0,\ldots,\overline{l}^{l_t-1},
\overline{s}^0,\ldots,\overline{s}^{s_t-1}\}\;.\]
By Definition~\ref{def:abstract_transformers},
\[\pl_{1,2} = \exists_T\big(
\pl_1[\outputv{v}\mapsto\overline{v}\;|\;\overline{v}\in T] \cup
\pl_2[\inputv{v}\mapsto\overline{v}\;|\;\overline{v}\in T]
\big)\;.\]
As $\rho$ is a model of $\pl_{1,2}$, there exists an assignment $\rho'$
that coincides with $\rho$ on every variable not in $T$ and
which is such that
\begin{equation}\label{equation-composition-1}
  \rho'\models\pl_1[\outputv{v}\mapsto\overline{v}\;|\;\overline{v}\in T]
  \quad\text{and}\quad
  \rho'\models\pl_2[\inputv{v}\mapsto\overline{v}\;|\;\overline{v}\in T]
  \;.
\end{equation}
Consider the following facts and definitions.
\begin{itemize}
\item The state $\sigma$ is compatible with $\delta_1;\delta_2$,
  hence it is compatible with $\delta_1$.
\item We define the assignment $\rho_1$ as:
  \[\begin{array}{rcl}
    \rho_1 & = & [\inputv{l}^k\mapsto \rho'(\inputv{l}^k)
    \;|\; 0\le k < l_i]\\[1ex]
    & \cup & [\inputv{s}^k\mapsto \rho'(\inputv{s}^k)
    \;|\; 0\le k < s_i]\\[1ex]
    & \cup & [\outputv{v}\mapsto\rho'(\overline{v})\;|\;
    \overline{v}\in T]\;.
  \end{array}\]
  By~(\ref{equation-composition-1}), $\rho_1\models \pl_1$, hence
  $\rho_1$ is a model of $\pl_1$. Moreover, $\sigma$ is compatible with
  $\delta_1$ and $\inputv{len}(\sigma) = \inputv{\rho}$ with
  $\inputv{\rho} = \inputv{\rho}_1$.
  As $\pl_1\models\delta_1$, then $\delta_1(\sigma)$ is defined and
  $\outputv{len}(\delta_1(\sigma)) = \outputv{\rho}_1$.
\item We define the assignment $\rho_2$ as:
  \[\begin{array}{rcl}
    \rho_2 & = & [\inputv{v}\mapsto\rho'(\overline{v})\;|\;
    \overline{v}\in T]\\[1ex]
    & \cup & [\outputv{l}^k\mapsto \rho'(\outputv{l}^k)
    \;|\; 0\le k < l_o]\\[1ex]
    & \cup & [\outputv{s}^k\mapsto \rho'(\outputv{s}^k)
    \;|\; 0\le k < s_o]\;.
  \end{array}\]
  By~(\ref{equation-composition-1}), $\rho_2\models \pl_2$, hence
  $\rho_2$ is a model of $\pl_2$. Moreover, $\delta_1(\sigma)$ is
  compatible with $\delta_2$ (because $\delta_1$ is compatible with
  $\delta_2$) and $\outputv{len}(\delta_1(\sigma)) = \outputv{\rho}_1$.
  By definition of $\rho_1$ and $\rho_2$, for each variable
  $\overline{v}$ in $T$ we have
  $\rho_1(\outputv{v})= \rho_2(\inputv{v})$. Hence,
  $\outputv{len}(\delta_1(\sigma)) = \outputv{\rho}_1$ implies
  that $\inputv{len}(\delta_1(\sigma)) = \inputv{\rho}_2$.
  As $\pl_2\models\delta_2$, then $\delta_2(\delta_1(\sigma))$ is
  defined and we have
  $\outputv{len}(\delta_2(\delta_1(\sigma))) = \outputv{\rho}_2$
  with $\outputv{\rho}_2 = \outputv{\rho}$.
\end{itemize}
Consequently, $(\delta_1;\delta_2)(\sigma)$ is defined and
$\outputv{len}((\delta_1;\delta_2)(\sigma)) = \outputv{\rho}$.

%%%%%%%%
\subsection{Proposition~\ref{proposition:instructions_exactness}}
%%%%%%%%
Let $q$ be a program point and $\#l,\#s$ be the number of local
variables and stack elements at $q$.
For any $L,S\subseteq\nat$, we let
\begin{align*}
  \mathit{Unchanged_q(L,S)} &= \{\inputv{l}^i=\outputv{l}^i\mid i\in L\}\\
  &\cup\{\inputv{s}^i=\outputv{s}^i\mid i\in S\}\\
  &\cup\{\inputv{s}^i=\inputv{s}^j \mid 0\le i,\ j<\#s,\\
  & \hspace{2cm}s^i\text{ is an alias of }s^j\text{ at }q\}\\
  &\cup\{\inputv{s}^i=\inputv{l}^j \mid
  0\le i<\#s,\ 0\le j<\#l,\\
  & \hspace{2cm}s^i\text{ is an alias of }l^j\text{ at }q\}\\
  &\cup\{\inputv{l}^i=\inputv{l}^j \mid 0\le i,\ j<\#l,\\
  & \hspace{2cm}l^i\text{ is an alias of }l^j\text{ at }q\}\\
  &\cup\{\inputv{s}^i\ge 0\mid 0\le i<\#s,\\
  & \hspace{0.5cm}\text{$s^i$ does not have integer type at }q\}\\
  &\cup\{\inputv{l}^i\ge 0\mid 0\le i<\#l,\\
  & \hspace{0.5cm}\text{$l^i$ does not have integer type at }q\}~.
\end{align*}
For any $l,s\in\nat$, we let
\begin{align*}
  \mathit{Unchanged}_q(l,s) = 
  \mathit{Unchanged}_q( & \{0,\ldots,l-1\},\\
  & \{0,\ldots,s-1\}).
\end{align*}

Let $\mathsf{ins}$ be an instruction in
$\correctins\setminus \{\mathsf{call}\}$.
We have to prove that $\mathit{ins}_q^{\mathbb{PL}}\models\mathit{ins}_q$.
Hence, consider any model $\rho$ of $\mathit{ins}_q^{\mathbb{PL}}$
and any state $\sigma$ compatible with $\mathit{ins}_q$, and suppose
that $\inputv{\len}(\sigma) = \inputv{\rho}$. We have to prove that
$\mathit{ins}_q(\sigma)$ is defined and that 
$\outputv{\len}(\mathit{ins}_q(\sigma)) = \outputv{\rho}$.
\begin{itemize}
\item Suppose that $\mathsf{ins} = \mathsf{const}\ c$. Then,
  \begin{align*}
    \mathit{ins}_q^\mathbb{PL} & =
    \begin{cases}
      \mathit{Unchanged}_q(\#l,\#s)\cup\{c=\outputv{s}^{\#s}\} &
      \text{if $c\in\mathbb{Z}$}\\
      \mathit{Unchanged}_q(\#l,\#s)\cup\{0=\outputv{s}^{\#s}\} &
      \text{if $c=\mathtt{null}$}
    \end{cases}\\
    \mathit{ins}_q & = \lambda\langle l\sep s\sep\mu\rangle.
    \langle l\sep c::s\sep\mu\rangle\;.
  \end{align*}
  Note that $\mathit{ins}_q(\sigma)$ is defined because 
  $\mathit{ins}_q$ is defined for any state that is compatible
  with it. Without loss of generality, suppose that $\sigma$
  has the form $\langle l\sep s\sep \mu\rangle$.
  \begin{itemize}
  \item Let $j\in\{0,\ldots,\#l-1\}$ and $k\in\{0,\ldots,\#s-1\}$.
    By definition of $\mathit{ins}_q$, we have
    $\outputv{\len}(\mathit{ins}_q(\sigma))(\outputv{l}^j) = 
    \inputv{\len}(\sigma)(\inputv{l}^j)$ and
    $\outputv{\len}(\mathit{ins}_q(\sigma))(\outputv{s}^k) = 
    \inputv{\len}(\sigma)(\inputv{s}^k)$; as 
    $\inputv{\len}(\sigma) = \inputv{\rho}$, 
    we have $\inputv{\len}(\sigma)(\inputv{l}^j) =
    \inputv{\rho}(\inputv{l}^j)$ and we have
    $\inputv{\len}(\sigma)(\inputv{s}^k) =
    \inputv{\rho}(\inputv{s}^k)$; moreover, as
    $\rho$ is a model of $\mathit{ins}_q^\mathbb{PL}$, with
    $\mathit{Unchanged}_q(\#l,\#s)\subseteq \mathit{ins}_q^\mathbb{PL}$,
    we have $\inputv{\rho}(\inputv{l}^j) = \outputv{\rho}(\outputv{l}^j)$
    and $\inputv{\rho}(\inputv{s}^k) = \outputv{\rho}(\outputv{s}^k)$.
    Therefore, 
    \[\outputv{\len}(\mathit{ins}_q(\sigma))(\outputv{l}^j) = 
    \outputv{\rho}(\outputv{l}^j)\] 
    and
    \[\outputv{\len}(\mathit{ins}_q(\sigma))(\outputv{s}^k) = 
    \outputv{\rho}(\outputv{s}^k)\;.\]
  \item Suppose that $c\in\mathbb{Z}$. By definition of
    $\mathit{ins}_q$, we have
    $\outputv{\len}(\mathit{ins}_q(\sigma))(\outputv{s}^{\#s}) = 
    \len(c, \mu) = c$. As $\rho$ is a model of
    $\mathit{ins}_q^\mathbb{PL}$, with
    $\{c=\outputv{s}^{\#s}\} \subseteq \mathit{ins}_q^\mathbb{PL}$,
    we have $\outputv{\rho}(\outputv{s}^{\#s}) = c$. Hence,
    \[\outputv{\len}(\mathit{ins}_q(\sigma))(\outputv{s}^{\#s}) = 
    \outputv{\rho}(\outputv{s}^{\#s})\;.\]
  \item Suppose that $c = \mathtt{null}$. By definition of
    $\mathit{ins}_q$, % we have
    $\outputv{\len}(\mathit{ins}_q(\sigma))(\outputv{s}^{\#s}) = 
    \len(c, \mu) = \len(\mathtt{null}, \mu) = 0$. As
    $\rho$ is a model of $\mathit{ins}_q^\mathbb{PL}$, with
    $\{0=\outputv{s}^{\#s}\} \subseteq \mathit{ins}_q^\mathbb{PL}$,
    we have $\outputv{\rho}(\outputv{s}^{\#s}) = 0$. Hence,
    \[\outputv{\len}(\mathit{ins}_q(\sigma))(\outputv{s}^{\#s}) = 
    \outputv{\rho}(\outputv{s}^{\#s})\;.\]
  \end{itemize}
  Consequently, we have $\outputv{\len}(\mathit{ins}_q(\sigma)) =
  \outputv{\rho}$.
\item Suppose that $\mathsf{ins} = \mathsf{dup}$. Then,
  \begin{align*}
    \mathit{ins}_q^\mathbb{PL} & = \mathit{Unchanged}_q(\#l,\#s)\cup
    \{\inputv{s}^{\#s-1}=\outputv{s}^{\#s}\}\\
    \mathit{ins}_q & = \lambda\langle l\sep\mathit{top}::s\sep\mu\rangle.
    \langle l\sep\mathit{top}::\mathit{top}::s\sep\mu\rangle\;.
  \end{align*}
  Note that $\#s\geq 1$ and that
  $\mathit{ins}_q(\sigma)$ is defined because 
  $\mathit{ins}_q$ is defined for any state that is compatible
  with it. Without loss of generality, suppose that $\sigma$
  has the form $\langle l\sep \mathit{top}::s\sep \mu\rangle$.
  \begin{itemize}
  \item Let $j\in\{0,\ldots,\#l-1\}$ and $k\in\{0,\ldots,\#s-1\}$.
    By definition of $\mathit{ins}_q$, we have
    $\outputv{\len}(\mathit{ins}_q(\sigma))(\outputv{l}^j) = 
    \inputv{\len}(\sigma)(\inputv{l}^j)$ and
    $\outputv{\len}(\mathit{ins}_q(\sigma))(\outputv{s}^k) = 
    \inputv{\len}(\sigma)(\inputv{s}^k)$; as 
    $\inputv{\len}(\sigma) = \inputv{\rho}$, 
    we have $\inputv{\len}(\sigma)(\inputv{l}^j) =
    \inputv{\rho}(\inputv{l}^j)$ and we have
    $\inputv{\len}(\sigma)(\inputv{s}^k) =
    \inputv{\rho}(\inputv{s}^k)$; moreover, as
    $\rho$ is a model of $\mathit{ins}_q^\mathbb{PL}$, with
    $\mathit{Unchanged}_q(\#l,\#s)\subseteq \mathit{ins}_q^\mathbb{PL}$,
    we have $\inputv{\rho}(\inputv{l}^j) = \outputv{\rho}(\outputv{l}^j)$
    and $\inputv{\rho}(\inputv{s}^k) = \outputv{\rho}(\outputv{s}^k)$.
    Therefore, 
    \[\outputv{\len}(\mathit{ins}_q(\sigma))(\outputv{l}^j) = 
    \outputv{\rho}(\outputv{l}^j)\]
    and
    \[\outputv{\len}(\mathit{ins}_q(\sigma))(\outputv{s}^k) = 
    \outputv{\rho}(\outputv{s}^k)\;.\]
  \item By definition of $\mathit{ins}_q$, we have
    $\outputv{\len}(\mathit{ins}_q(\sigma))(\outputv{s}^{\#s}) = 
    \len(\mathit{top}, \mu)$ with $\len(\mathit{top}, \mu) = 
    \inputv{\len}(\sigma)(\inputv{s}^{\#s-1})$.
    As $\inputv{\len}(\sigma) = \inputv{\rho}$, 
    we have $\inputv{\len}(\sigma)(\inputv{s}^{\#s-1}) =
    \inputv{\rho}(\inputv{s}^{\#s-1})$.
    As $\rho$ is a model of $\mathit{ins}_q^\mathbb{PL}$, with
    $\{\inputv{s}^{\#s-1}=\outputv{s}^{\#s}\} \subseteq
    \mathit{ins}_q^\mathbb{PL}$,
    we have
    $\inputv{\rho}(\inputv{s}^{\#s-1}) = \outputv{\rho}(\outputv{s}^{\#s})$.
    Hence,
    \[\outputv{\len}(\mathit{ins}_q(\sigma))(\outputv{s}^{\#s}) = 
    \outputv{\rho}(\outputv{s}^{\#s})\;.\]
  \end{itemize}
  Consequently, we have $\outputv{\len}(\mathit{ins}_q(\sigma)) =
  \outputv{\rho}$.
\item Suppose that $\mathsf{ins} = \mathsf{new}\ \kappa$. Then,
  \begin{align*}
    \mathit{ins}_q^\mathbb{PL} & = \mathit{Unchanged}_q(\#l,\#s)\cup
    \{1=\outputv{s}^{\#s}\}\\
    \mathit{ins_q} & = \lambda\langle l\sep s\sep\mu\rangle.
    \langle l\sep\ell::s\sep\mu[\ell\mapsto o]\rangle\\
    & \text{where $\ell$ is a fresh location}\\
    & \text{and $o$ is an object of class $\kappa$}\\
    & \text{whose fields hold $0$ or $\mathtt{null}$.}
  \end{align*}
  Note that $\mathit{ins}_q(\sigma)$ is defined because we assume
  that $\mathit{ins}_q$ is a total map%
  \footnote{This is true only if we assume that the
    system has infinite memory. Termination because of
    out of memory is not really termination from our point of view.}
  defined for any state that is compatible
  with it. Without loss of generality, suppose that $\sigma$
  has the form $\langle l\sep s\sep \mu\rangle$.
  \begin{itemize}
  \item Let $j\in\{0,\ldots,\#l-1\}$ and $k\in\{0,\ldots,\#s-1\}$.
    By definition of $\mathit{ins}_q$, we have
    $\outputv{\len}(\mathit{ins}_q(\sigma))(\outputv{l}^j) = 
    \inputv{\len}(\sigma)(\inputv{l}^j)$ and
    $\outputv{\len}(\mathit{ins}_q(\sigma))(\outputv{s}^k) = 
    \inputv{\len}(\sigma)(\inputv{s}^k)$; as 
    $\inputv{\len}(\sigma) = \inputv{\rho}$, 
    we have $\inputv{\len}(\sigma)(\inputv{l}^j) =
    \inputv{\rho}(\inputv{l}^j)$ and we have
    $\inputv{\len}(\sigma)(\inputv{s}^k) =
    \inputv{\rho}(\inputv{s}^k)$; moreover, as
    $\rho$ is a model of $\mathit{ins}_q^\mathbb{PL}$, with
    $\mathit{Unchanged}_q(\#l,\#s)\subseteq \mathit{ins}_q^\mathbb{PL}$,
    we have $\inputv{\rho}(\inputv{l}^j) = \outputv{\rho}(\outputv{l}^j)$
    and $\inputv{\rho}(\inputv{s}^k) = \outputv{\rho}(\outputv{s}^k)$.
    Therefore, 
    \[\outputv{\len}(\mathit{ins}_q(\sigma))(\outputv{l}^j) = 
    \outputv{\rho}(\outputv{l}^j)\]
    and
    \[\outputv{\len}(\mathit{ins}_q(\sigma))(\outputv{s}^k) = 
    \outputv{\rho}(\outputv{s}^k)\;.\]
  \item By definition of $\mathit{ins}_q$, we have
    $\outputv{\len}(\mathit{ins}_q(\sigma))(\outputv{s}^{\#s}) = 
    \len(\ell, \mu) = 1$.
    As $\rho$ is a model of $\mathit{ins}_q^\mathbb{PL}$, with
    $\{1 =\outputv{s}^{\#s}\} \subseteq
    \mathit{ins}_q^\mathbb{PL}$,
    we have
    $1 = \outputv{\rho}(\outputv{s}^{\#s})$.
    Hence,
    \[\outputv{\len}(\mathit{ins}_q(\sigma))(\outputv{s}^{\#s}) = 
    \outputv{\rho}(\outputv{s}^{\#s})\;.\]
  \end{itemize}
  Consequently, we have $\outputv{\len}(\mathit{ins}_q(\sigma)) =
  \outputv{\rho}$.
\item Suppose that $\mathsf{ins} = \mathsf{load}\ i$. Then,
  \begin{align*}
    \mathit{ins}_q^\mathbb{PL} & = \mathit{Unchanged}_q(\#l,\#s)\cup
    \{\inputv{l}^i=\outputv{s}^{\#s}\}\\
    \mathit{ins}_q & =\lambda\langle l\sep s\sep\mu\rangle.
    \langle l\sep l^i::s\sep\mu\rangle \; .
  \end{align*}
  Note that $\mathit{ins}_q(\sigma)$ is defined because 
  $\mathit{ins}_q$ is defined for any state that is compatible
  with it. Without loss of generality, suppose that $\sigma$
  has the form $\langle l\sep s\sep \mu\rangle$.
  \begin{itemize}
  \item Let $j\in\{0,\ldots,\#l-1\}$ and $k\in\{0,\ldots,\#s-1\}$.
    By definition of $\mathit{ins}_q$, we have
    $\outputv{\len}(\mathit{ins}_q(\sigma))(\outputv{l}^j) = 
    \inputv{\len}(\sigma)(\inputv{l}^j)$ and
    $\outputv{\len}(\mathit{ins}_q(\sigma))(\outputv{s}^k) = 
    \inputv{\len}(\sigma)(\inputv{s}^k)$; as 
    $\inputv{\len}(\sigma) = \inputv{\rho}$, 
    we have $\inputv{\len}(\sigma)(\inputv{l}^j) =
    \inputv{\rho}(\inputv{l}^j)$ and we have
    $\inputv{\len}(\sigma)(\inputv{s}^k) =
    \inputv{\rho}(\inputv{s}^k)$; moreover, as
    $\rho$ is a model of $\mathit{ins}_q^\mathbb{PL}$, with
    $\mathit{Unchanged}_q(\#l,\#s)\subseteq \mathit{ins}_q^\mathbb{PL}$,
    we have $\inputv{\rho}(\inputv{l}^j) = \outputv{\rho}(\outputv{l}^j)$
    and $\inputv{\rho}(\inputv{s}^k) = \outputv{\rho}(\outputv{s}^k)$.
    Therefore, 
    \[\outputv{\len}(\mathit{ins}_q(\sigma))(\outputv{l}^j) = 
    \outputv{\rho}(\outputv{l}^j)\]
    and
    \[\outputv{\len}(\mathit{ins}_q(\sigma))(\outputv{s}^k) = 
    \outputv{\rho}(\outputv{s}^k)\;.\]
  \item By definition of $\mathit{ins}_q$, we have
    $\outputv{\len}(\mathit{ins}_q(\sigma))(\outputv{s}^{\#s}) = 
    \len(l^i, \mu) =
    \inputv{\len}(\sigma)(\inputv{l}^i)$.
    As $\inputv{\len}(\sigma) = \inputv{\rho}$, 
    we have $\inputv{\len}(\sigma)(\inputv{l}^i) =
    \inputv{\rho}(\inputv{l}^i)$.
    As $\rho$ is a model of $\mathit{ins}_q^\mathbb{PL}$, with
    $\{\inputv{l}^i=\outputv{s}^{\#s}\} \subseteq
    \mathit{ins}_q^\mathbb{PL}$,
    we have
    $\inputv{\rho}(\inputv{l}^i) = \outputv{\rho}(\outputv{s}^{\#s})$.
    Hence,
    \[\outputv{\len}(\mathit{ins}_q(\sigma))(\outputv{s}^{\#s}) = 
    \outputv{\rho}(\outputv{s}^{\#s})\;.\]
  \end{itemize}
  Consequently, we have $\outputv{\len}(\mathit{ins}_q(\sigma)) =
  \outputv{\rho}$.
\item Suppose that $\mathsf{ins} = \mathsf{store}\ i$. Then,
  \begin{align*}
    \mathit{ins}_q^\mathbb{PL} & = \mathit{Unchanged}_q
    (\{0,\ldots,\#l-1\}\setminus i,\\
    & \hspace{2.3cm}\{0,\ldots,\#s-2\})\cup
    \{\inputv{s}^{\#s-1}=\outputv{l}^i\}\\
    \mathit{ins}_q & = \lambda\langle l\sep\mathit{top}::s\sep\mu\rangle.
    \langle l[i\mapsto\mathit{top}]\sep s\sep\mu\rangle \; .
  \end{align*}
  Note that $\#s\geq 1$, that $0\le i \le \#l-1$ and that
  $\mathit{ins}_q(\sigma)$ is defined because 
  $\mathit{ins}_q$ is defined for any state that is compatible
  with it. Without loss of generality, suppose that $\sigma$
  has the form $\langle l\sep \mathit{top}::s\sep \mu\rangle$.
  \begin{itemize}
  \item Let $j\in\{0,\ldots,\#l-1\}\setminus i$ and
    $k\in\{0,\ldots,\#s-2\}$.
    By definition of $\mathit{ins}_q$, we have
    $\outputv{\len}(\mathit{ins}_q(\sigma))(\outputv{l}^j) = 
    \inputv{\len}(\sigma)(\inputv{l}^j)$ and
    $\outputv{\len}(\mathit{ins}_q(\sigma))(\outputv{s}^k) = 
    \inputv{\len}(\sigma)(\inputv{s}^k)$; as 
    $\inputv{\len}(\sigma) = \inputv{\rho}$, 
    we have $\inputv{\len}(\sigma)(\inputv{l}^j) =
    \inputv{\rho}(\inputv{l}^j)$ and we have
    $\inputv{\len}(\sigma)(\inputv{s}^k) =
    \inputv{\rho}(\inputv{s}^k)$; as
    $\rho$ is a model of $\mathit{ins}_q^\mathbb{PL}$, with
    $\mathit{Unchanged}_q(\{0,\ldots,\#l-1\}\setminus i,\{0,\ldots,\#s-2\})
    \subseteq \mathit{ins}_q^\mathbb{PL}$,
    we have $\inputv{\rho}(\inputv{l}^j) = \outputv{\rho}(\outputv{l}^j)$
    and $\inputv{\rho}(\inputv{s}^k) = \outputv{\rho}(\outputv{s}^k)$.
    Therefore, 
    \[\outputv{\len}(\mathit{ins}_q(\sigma))(\outputv{l}^j) = 
    \outputv{\rho}(\outputv{l}^j)\]
    and
    \[\outputv{\len}(\mathit{ins}_q(\sigma))(\outputv{s}^k) = 
    \outputv{\rho}(\outputv{s}^k)\;.\]
  \item By definition of $\mathit{ins}_q$, % we have
    $\outputv{\len}(\mathit{ins}_q(\sigma))(\outputv{l}^i) = 
    \len(\mathit{top}, \mu)$ with $\len(\mathit{top}, \mu) = 
    \inputv{\len}(\sigma)(\inputv{s}^{\#s-1})$.
    As $\inputv{\len}(\sigma) = \inputv{\rho}$, 
    we have $\inputv{\len}(\sigma)(\inputv{s}^{\#s-1}) =
    \inputv{\rho}(\inputv{s}^{\#s-1})$.
    As $\rho$ is a model of $\mathit{ins}_q^\mathbb{PL}$, with
    $\{\inputv{s}^{\#s-1}=\outputv{l}^i\} \subseteq
    \mathit{ins}_q^\mathbb{PL}$,
    we have
    $\inputv{\rho}(\inputv{s}^{\#s-1}) = \outputv{\rho}(\outputv{l}^i)$.
    Hence,
    \[\outputv{\len}(\mathit{ins}_q(\sigma))(\outputv{l}^i) = 
    \outputv{\rho}(\outputv{l}^i)\;.\]
  \end{itemize}
  Consequently, we have $\outputv{\len}(\mathit{ins}_q(\sigma)) =
  \outputv{\rho}$.
\item Suppose that $\mathsf{ins} = \mathsf{add}$. Then,
  \begin{align*}
    \mathit{ins}_q^\mathbb{PL} & = \mathit{Unchanged}_q(\#l,\#s-2)\\
    & \cup
    \{\inputv{s}^{\#s-2}+\inputv{s}^{\#s-1}=\outputv{s}^{\#s-2}\}\\
    \mathit{ins}_q & = \lambda\langle l\sep x::y::s\sep\mu\rangle.
    \langle l\sep (x+y)::s\sep\mu\rangle \; .
  \end{align*}
  Note that $\#s\geq 2$. Without loss of generality, suppose that
  $\sigma$ has the form $\langle l\sep x::y::s\sep \mu\rangle$.
  As $\sigma$ is compatible with $\mathit{ins}_q$, 
  $x$ and $y$ are integer values. Moreover,
  $\mathit{ins}_q(\sigma)$ is defined because 
  $\mathit{ins}_q$ is defined for any state that is compatible
  with it. 
  \begin{itemize}
  \item Let $j\in\{0,\ldots,\#l-1\}$ and
    $k\in\{0,\ldots,\#s-3\}$.
    By definition of $\mathit{ins}_q$, we have
    $\outputv{\len}(\mathit{ins}_q(\sigma))(\outputv{l}^j) = 
    \inputv{\len}(\sigma)(\inputv{l}^j)$ and
    $\outputv{\len}(\mathit{ins}_q(\sigma))(\outputv{s}^k) = 
    \inputv{\len}(\sigma)(\inputv{s}^k)$; as 
    $\inputv{\len}(\sigma) = \inputv{\rho}$, 
    we have $\inputv{\len}(\sigma)(\inputv{l}^j) =
    \inputv{\rho}(\inputv{l}^j)$ and we have
    $\inputv{\len}(\sigma)(\inputv{s}^k) =
    \inputv{\rho}(\inputv{s}^k)$; as
    $\rho$ is a model of $\mathit{ins}_q^\mathbb{PL}$, with
    $\mathit{Unchanged}_q(\#l,\#s-2)
    \subseteq \mathit{ins}_q^\mathbb{PL}$,
    we have $\inputv{\rho}(\inputv{l}^j) = \outputv{\rho}(\outputv{l}^j)$
    and $\inputv{\rho}(\inputv{s}^k) = \outputv{\rho}(\outputv{s}^k)$.
    Therefore, 
    \[\outputv{\len}(\mathit{ins}_q(\sigma))(\outputv{l}^j) = 
    \outputv{\rho}(\outputv{l}^j)\]
    and
    \[\outputv{\len}(\mathit{ins}_q(\sigma))(\outputv{s}^k) = 
    \outputv{\rho}(\outputv{s}^k)\;.\]
  \item By definition of $\mathit{ins}_q$, we have
    \[\outputv{\len}(\mathit{ins}_q(\sigma))(\outputv{s}^{\#s-2}) = 
    \len(x+y, \mu)\]
    with
    \begin{align*}
      \len(x+y, \mu) & = x+y\\
      & = \len(x,\mu) + \len(y, \mu)\\
      & = \inputv{\len}(\sigma)(\inputv{s}^{\#s-1}) +
      \inputv{\len}(\sigma)(\inputv{s}^{\#s-2})\;.
    \end{align*}
    As $\inputv{\len}(\sigma) = \inputv{\rho}$, 
    we have $\inputv{\len}(\sigma)(\inputv{s}^{\#s-1}) =
    \inputv{\rho}(\inputv{s}^{\#s-1})$ and
    $\inputv{\len}(\sigma)(\inputv{s}^{\#s-2}) =
    \inputv{\rho}(\inputv{s}^{\#s-2})$.
    As $\rho$ is a model of $\mathit{ins}_q^\mathbb{PL}$, with
    $\{\inputv{s}^{\#s-2}+\inputv{s}^{\#s-1}=\outputv{s}^{\#s-2}\} \subseteq
    \mathit{ins}_q^\mathbb{PL}$,
    we have
    $\inputv{\rho}(\inputv{s}^{\#s-2}) +
    \inputv{\rho}(\inputv{s}^{\#s-1}) = \outputv{\rho}(\outputv{s}^{\#s-2})$.
    Hence,
    \[\outputv{\len}(\mathit{ins}_q(\sigma))(\outputv{s}^{\#s-2}) = 
    \outputv{\rho}(\outputv{s}^{\#s-2})\;.\]
  \end{itemize}
  Consequently, we have $\outputv{\len}(\mathit{ins}_q(\sigma)) =
  \outputv{\rho}$.
\item Suppose that $\mathsf{ins} = \mathsf{putfield}\ f$
  where $f$ has integer type.
  Then,
  \begin{align*}
    \mathit{ins}_q^\mathbb{PL} & =
    \mathit{Unchanged}_q(\#l,\#s-2) \cup \{\inputv{s}^{\#s-2} \geq 1\}\\
    \mathit{ins}_q & = 
    \lambda\langle l\sep v::\ell::s\sep\mu\rangle.
    \begin{cases}
      \langle l\sep s\sep\mu[\ell\mapsto\mu(\ell)[f\mapsto v]]\rangle\\
      \hspace{1.5cm}\text{if $\ell\neq\mathtt{null}$}\\
      \mathit{undefined}\quad\text{otherwise.}
    \end{cases}
  \end{align*}
  Note that $\#s\geq 2$. Without loss of generality, suppose that
  $\sigma$ has the form
  $\langle l\sep v::\ell::s\sep \mu\rangle$.
  As $\rho$ is a model of $\mathit{ins}_q^\mathbb{PL}$, with
  $\{\inputv{s}^{\#s-2}\geq 1\} \subseteq \mathit{ins}_q^\mathbb{PL}$,
  we have $\inputv{\rho}(\inputv{s}^{\#s-2}) \geq 1$.
  As $\inputv{\len}(\sigma) = \inputv{\rho}$, then we have
  $\inputv{\len}(\sigma)(\inputv{s}^{\#s-2}) \geq 1$,
  with $\inputv{\len}(\sigma)(\inputv{s}^{\#s-2}) =
  \len(\ell,\mu)$. Hence, $\len(\ell,\mu) \geq 1$.
  As $\sigma$ is compatible with $\mathit{ins}_q$, $\ell$ does
  not have integer type. So, $\len(\ell,\mu) \geq 1$
  implies that $\ell \neq \mathtt{null}$. Consequently,
  $\mathit{ins}_q(\sigma)$ is defined.
  
  We let $\mu'$ denote the memory 
  $\mu[\ell\mapsto\mu(\ell)[f\mapsto v]]$.
  Note that $\dom(\mu)=\dom(\mu')$.
  \begin{claim}
    For any $i\in\integers$, we have
    \[\len(i,\mu)=\len(i,\mu')\;.\]
    For any $\ell\ell\in\dom(\mu)$, we have
    \[\len(\ell\ell,\mu)=\len(\ell\ell,\mu')\;.\]
  \end{claim}
  \begin{proof}
    By Definition~24 in~\cite{SpotoMP10},
    for any $i\in\integers$ we have $\len(i,\mu) = i$ and
    $\len(i,\mu') = i$. Hence, $\len(i,\mu)=\len(i,\mu')$.
    
    Let $\ell\ell\in\dom(\mu)$. Note that $\mu'$ coincides with
    $\mu$, except, possibly, on the value of field $f$ of
    objects $\mu(\ell)$ and $\mu'(\ell)$.
    Field $f$ has integer type; by Definition~24 in~\cite{SpotoMP10},
    the path-length of a location does not depend on the value 
    of the fields with integer type of the objects in memory.
    Hence, $\len(\ell\ell,\mu)=\len(\ell\ell,\mu')$.
    \qed
  \end{proof}
  Let $j\in\{0,\ldots,\#l-1\}$ and $k\in\{0,\ldots,\#s-3\}$.
  By definition of $\mathit{ins}_q$ and by the claim above,
  \[\outputv{\len}(\mathit{ins}_q(\sigma))(\outputv{l}^j) = 
  \inputv{\len}(\sigma)(\inputv{l}^j)\] and
  \[\outputv{\len}(\mathit{ins}_q(\sigma))(\outputv{s}^k) = 
  \inputv{\len}(\sigma)(\inputv{s}^k)\;.\]
  As 
  $\inputv{\len}(\sigma) = \inputv{\rho}$, 
  we have
  \[\inputv{\len}(\sigma)(\inputv{l}^j) =
  \inputv{\rho}(\inputv{l}^j) \quad\text{and}\quad
  \inputv{\len}(\sigma)(\inputv{s}^k) =
  \inputv{\rho}(\inputv{s}^k)\;.\]
  Moreover, as $\rho$ is a model of $\mathit{ins}_q^\mathbb{PL}$, with
  \[\mathit{Unchanged}_q(\#l,\#s-2)\subseteq \mathit{ins}_q^\mathbb{PL}\]
  we have $\inputv{\rho}(\inputv{l}^j) = \outputv{\rho}(\outputv{l}^j)$
  and $\inputv{\rho}(\inputv{s}^k) = \outputv{\rho}(\outputv{s}^k)$.
  Therefore, 
  \[\outputv{\len}(\mathit{ins}_q(\sigma))(\outputv{l}^j) = 
  \outputv{\rho}(\outputv{l}^j)\]
  and
  \[\outputv{\len}(\mathit{ins}_q(\sigma))(\outputv{s}^k) = 
  \outputv{\rho}(\outputv{s}^k)\;.\]
  
  Consequently, we have $\outputv{\len}(\mathit{ins}_q(\sigma)) =
  \outputv{\rho}$.
\item Suppose that $\mathsf{ins} = \mathsf{ifeq\ of\ type}\ t$.
  Then,
  \begin{align*}
    \mathit{ins}_q^\mathbb{PL} & =
    \mathit{Unchanged}_q(\#l,\#s-1)\cup \{\inputv{s}^{\#s-1}=0\}\\
    \mathit{ins}_q & = 
    \lambda\langle l\sep\mathit{top}::s\sep\mu\rangle.
    \begin{cases}
      \langle l\sep s\sep\mu\rangle\\
      \quad\text{if $\mathit{top}=0$ or $\mathit{top}=\mathtt{null}$}\\
      \mathit{undefined}\quad\text{otherwise.}
    \end{cases}
  \end{align*}
  Note that $\#s\geq 1$. Without loss of generality, suppose that
  $\sigma$ has the form $\langle l\sep \mathit{top}::s\sep \mu\rangle$.
  As $\rho$ is a model of $\mathit{ins}_q^\mathbb{PL}$, with
  $\{\inputv{s}^{\#s-1}=0\} \subseteq \mathit{ins}_q^\mathbb{PL}$,
  we have $\inputv{\rho}(\inputv{s}^{\#s-1}) = 0$.
  As $\inputv{\len}(\sigma) = \inputv{\rho}$, then we have
  $\inputv{\len}(\sigma)(\inputv{s}^{\#s-1}) = 0$,
  with $\inputv{\len}(\sigma)(\inputv{s}^{\#s-1}) = \len(\mathit{top},\mu)$.
  Hence, $\len(\mathit{top},\mu) = 0$.
  \begin{itemize}
  \item Suppose that $t = \mathtt{int}$.
    As $\sigma$ is compatible with $\mathit{ins}_q$, $\mathit{top}$
    has type $\mathtt{int}$, hence $\len(\mathit{top},\mu) = 0$
    implies that $\mathit{top} = 0$. Consequently,
    $\mathit{ins}_q(\sigma)$ is defined.
  \item Suppose that $t \neq \mathtt{int}$.
    As $\sigma$ is compatible with $\mathit{ins}_q$, $\mathit{top}$
    has type $t \neq \mathtt{int}$, hence $\len(\mathit{top},\mu) = 0$
    implies that $\mathit{top} = \mathtt{null}$. Consequently,
    $\mathit{ins}_q(\sigma)$ is defined.
  \end{itemize}
  Let $j\in\{0,\ldots,\#l-1\}$ and $k\in\{0,\ldots,\#s-2\}$.
  By definition of $\mathit{ins}_q$, we have
  $\outputv{\len}(\mathit{ins}_q(\sigma))(\outputv{l}^j) = 
  \inputv{\len}(\sigma)(\inputv{l}^j)$ and
  $\outputv{\len}(\mathit{ins}_q(\sigma))(\outputv{s}^k) = 
  \inputv{\len}(\sigma)(\inputv{s}^k)$; as 
  $\inputv{\len}(\sigma) = \inputv{\rho}$, 
  we have $\inputv{\len}(\sigma)(\inputv{l}^j) =
  \inputv{\rho}(\inputv{l}^j)$ and
  $\inputv{\len}(\sigma)(\inputv{s}^k) =
  \inputv{\rho}(\inputv{s}^k)$; moreover, as
  $\rho$ is a model of $\mathit{ins}_q^\mathbb{PL}$, with
  \[\mathit{Unchanged}_q(\#l,\#s-1)\subseteq \mathit{ins}_q^\mathbb{PL}\]
  we have $\inputv{\rho}(\inputv{l}^j) = \outputv{\rho}(\outputv{l}^j)$
  and $\inputv{\rho}(\inputv{s}^k) = \outputv{\rho}(\outputv{s}^k)$.
  Therefore, 
  \[\outputv{\len}(\mathit{ins}_q(\sigma))(\outputv{l}^j) = 
  \outputv{\rho}(\outputv{l}^j)\]
  and
  \[\outputv{\len}(\mathit{ins}_q(\sigma))(\outputv{s}^k) = 
  \outputv{\rho}(\outputv{s}^k)\;.\]
  
  Consequently, we have $\outputv{\len}(\mathit{ins}_q(\sigma)) =
  \outputv{\rho}$.
\end{itemize}

\subsection{Existence of a Compatible State}
%%%%%%%%%%%%%%%%%%%%%%%%%%%%%%%%%%%%%%%%%%%%%%
An important step of our analysis consists in deducing the
non-termination of $P$ from that of $P_{\mathit{CLP}}$.
The idea consists in constructing an infinite execution of $P$
from an infinite derivation with $P_{\mathit{CLP}}$. Each step of
the infinite derivation consists of an atom $b(\mathit{vars})$,
where $\mathit{vars}$ are integer values, and we must be able to
transform this atom into a state whose path-length matches
$\mathit{vars}$.
%
%%% Proposition:
\begin{proposition}\label{proposition:instructions_exists_sigma}
  For any $\mathsf{ins} \in \correctins\setminus 
  \{\mathsf{call}\}$, any program point $q$ where $\mathsf{ins}$ 
  occurs and any model $\rho$ of $\mathit{ins}^{\mathbb{PL}}_q$,
  there exists a state $\sigma$ which is compatible with
  $\mathit{ins}_q$ and such that 
  $\inputv{\len}(\sigma) = \inputv{\rho}$.
\end{proposition}

As a memory $\mu$ is a mapping from locations to objects,
we let $\dom(\mu)$ denote the domain of $\mu$.
An update of $\mu$ is written as $\mu[\ell\mapsto o]$, where
the domain of $\mu$ may be enlarged (if $\ell\not\in\dom(\mu)$).

Without loss of generality, we assume that every class $\kappa$
in the program under analysis satisfies the following property:
for any integer $n\geq 1$, any memory $\mu$ and any location
$\ell\not\in\dom(\mu)$, there exists an object
$o$ instance of $\kappa$ such that $\len(\ell, \mu[\ell\mapsto o]) = n$.
If the program includes a class $\kappa $ that does not satisfy
this property, just add to $\kappa$ a dummy field of type $\kappa$.
The termination of the transformed program is equivalent to that
of the original one.

Let $\mathsf{ins}$ be an instruction in
$\correctins\setminus\{\mathsf{call}\}$.
Let $q$ be a program point where $\mathsf{ins}$ occurs
and $\#l,\#s$ be the number of local variables and stack elements
at $q$. Let $\rho$ be a model of $\mathit{ins}_q^{\mathbb{PL}}$.
Given the above assumption, we can construct a state
$\langle l\sep s\sep \mu\rangle$ in $\Sigma_{\#l,\#s}$
which is such that: for any $k$ in $\{0,\ldots,\#l-1\}$
(resp. in $\{0,\ldots,\#s-1\}$),
\begin{itemize}
\item if the $k$th local variable (resp. stack element)
  has integer type at $q$, then $l^k$ is $\rho(\inputv{l}^k)$
  (resp. $s^k$ is $\rho(\inputv{s}^k)$);
\item if the $k$th local variable (resp. stack element)
  has class type at $q$ and $\rho(\inputv{l}^k) = 0$
  (resp. $\rho(\inputv{s}^k) = 0$),
  then $l^k$ is $\mathtt{null}$
  (resp. $s^k$ is $\mathtt{null}$);
\item if the $k$th local variable (resp. stack element)
  has class type at $q$ and $\rho(\inputv{l}^k) \neq 0$
  (resp. $\rho(\inputv{s}^k) \neq 0$), 
  then $l^k$ (resp. $s^k$) is a location $\ell$
  which is such that $\len(\ell,\mu) = \rho(\inputv{l}^k)$
  (resp. $\len(\ell,\mu) = \rho(\inputv{s}^k)$).
\end{itemize}
Then, $\sigma$ is compatible with $\mathit{ins}_q$ and we have
$\inputv{\len}(\sigma) = \inputv{\rho}$.

%%%%%%%%
\subsection{Theorem~\ref{theorem:clp_completeness}}
%%%%%%%%
Let $J$ be a Java Virtual Machine, $P$ be a Java bytecode program
consisting of instructions in $\correctins\setminus \{\mathsf{call}\}$,
and $b$ be a block of $P$.
By Theorem~56 of \cite{SpotoMP10},
if $b(\mathit{vars})$ has only terminating computations
in $P_{\mathit{CLP}}$, for any fixed integer values for
$\mathit{vars}$, then all executions of $J$
started at $b$ terminate.

Let us prove that if all executions of $J$
started at $b$ terminate, then
$b(\mathit{vars})$ has only terminating computations
in $P_{\mathit{CLP}}$, for any fixed integer values for
$\mathit{vars}$. This is equivalent to proving that
if there exists some fixed integer values for
$\mathit{vars}$ such that $b(\mathit{vars})$ has an
infinite computation in $P_{\mathit{CLP}}$, then
there exists an execution of $J$ started at $b$ that
does not terminate.

Hence, suppose that for some fixed integer values
for $\mathit{vars}$ there exists an infinite computation
of $b(\mathit{vars})$ in $P_{\mathit{CLP}}$. Note that for
any block $b'$ of $P$, the unification of the CLP atom
$b'(\outputv{\mathit{vars}})$ with the atom
$b'(\inputv{\mathit{vars}})$ corresponds to the $;^{\mathbb{PL}}$
operation (renaming of the variables into new overlined variables
and existential quantification). 
Then, by definition of our specialised semantics,
there exists an infinite sequence of integer values
\[\mathit{vars} \ra_{c_0} \mathit{vars}_0 \ra_{c_1} 
\cdots \ra_{c_i} \mathit{vars}_i \ra_{c_{i+1}} \cdots\]
where
\[\begin{array}{rcl}
  b(\inputv{\mathit{vars}}) & \texttt{:-} &
  c_0,b_0(\outputv{\mathit{vars}}_0)\\
  b_0(\inputv{\mathit{vars}}_0) & \texttt{:-} &
  c_1,b_1(\outputv{\mathit{vars}}_1)\\
  & \vdots & \\
  b_i(\inputv{\mathit{vars}}_i) & \texttt{:-} &
  c_{i+1},b_{i+1}(\outputv{\mathit{vars}}_{i+1})\\
  & \vdots &
\end{array}\]
are clauses from $P_{\mathit{CLP}}$.
For each $i\in\nat$, $\mathit{vars}_i$ are integer values
for $\outputv{\mathit{vars}}_i$ and the assignment
\[\rho_i = [\inputv{\mathit{vars}} \mapsto \mathit{vars},
\outputv{\mathit{vars}}_i \mapsto \mathit{vars}_i]\]
is a model of $c_0 ;^{\mathbb{PL}} \dots ;^{\mathbb{PL}} c_i$.

By definition of $P_{\mathit{CLP}}$ (Definition~53
of~\cite{SpotoMP10}), we have
\[c_0 = \mathit{ins}_1^{\mathbb{PL}}
;^{\mathbb{PL}}\ldots;^{\mathbb{PL}}
\mathit{ins}_{w_0}^{\mathbb{PL}}\]
where $\mathsf{ins}_1$, \dots,$\mathsf{ins}_{w_0}$ are the
instructions occurring in block $b$ and, for each $i\geq 1$, % we have
\[c_i = \mathit{ins}_{w_{i-1}+1}^{\mathbb{PL}}
;^{\mathbb{PL}}\ldots;^{\mathbb{PL}}
\mathit{ins}_{w_i}^{\mathbb{PL}}\]
where $\mathsf{ins}_{w_{i-1}+1}$, \dots, $\mathsf{ins}_{w_i}$ are the
instructions occurring in block $b_{i-1}$.
We let
\[\delta_0=\mathit{ins}_1 ;\ldots;\mathit{ins}_{w_0}\]
and
\[\forall i\geq 1,
\delta_i=\mathit{ins}_{w_{i-1}+1};\ldots;\mathit{ins}_{w_i}\;.\]
As $P$ is a valid Java bytecode program, any $\mathit{ins}_k$
is compatible with its direct successor $\mathit{ins}_{k+1}$. Hence,
by Proposition~\ref{proposition:composition}
and Proposition~\ref{proposition:instructions_exactness}, 
for each $i\in\nat$ we have $c_i\models\delta_i$.

As $\rho_0$ is a model of $c_0$ and
$c_0 = \mathit{ins}_1^{\mathbb{PL}} ;^{\mathbb{PL}}\ldots;^{\mathbb{PL}}
\mathit{ins}_{w_0}^{\mathbb{PL}}$, there exists a model $\rho$ of
$\mathit{ins}_1^{\mathbb{PL}}$ which is such that
$\inputv{\rho} = \inputv{\rho}_0$.
By Proposition~\ref{proposition:instructions_exists_sigma}
there exists a state $\sigma$ compatible with $\mathit{ins}_1$
which is such that
$\inputv{\len}(\sigma) = \inputv{\rho} = \inputv{\rho}_0$.
Note that for each $i\in\nat$, the state $\sigma$ is compatible
with $\delta_0;\ldots;\delta_i$ because $\sigma$ is compatible
with $\mathit{ins}_1$ and $\mathit{ins}_1$ is the denotation
that is applied first in $\delta_0;\ldots;\delta_i$. Moreover,
for each $i\in\nat$ we have $\inputv{\rho}_0=\inputv{\rho}_i$ \ie
$\inputv{len}(\sigma)=\inputv{\rho}_i$.

Let $i\in\nat$. As $c_0\models\delta_0$, \dots, $c_i\models\delta_i$,
by Proposition~\ref{proposition:composition}
we have $(c_0;^{\mathbb{PL}}\ldots;^{\mathbb{PL}}c_i) \models 
(\delta_0;\ldots;\delta_i)$.
Moreover, $\rho_i$ is a model of $c_0;^{\mathbb{PL}}\ldots;^{\mathbb{PL}}c_i$,
the state $\sigma$ is compatible with $\delta_0;\ldots;\delta_i$
and $\inputv{len}(\sigma)=\inputv{\rho}_0=\inputv{\rho}_i$.
Hence, by Definition~\ref{definition:models}, 
$(\delta_0;\ldots;\delta_i)(\sigma)$ is defined.

Consequently, we have proved that 
$(\delta_0;\ldots;\delta_i)(\sigma)$ is defined for
any $i$ in $\nat$.
By the equivalence of the denotational and operational semantics
(Theorem~23 of~\cite{SpotoMP10}), there exists
an infinite operational execution of $J$ from block $b$
starting at state $\sigma$.

%%%%%%%%
\subsection{Theorem~\ref{theorem:non-term}}
%%%%%%%%
Let $J$ be a Java Virtual Machine, $P$ be a Java bytecode program
consisting of instructions in $\correctins$, and $b$ be a block of $P$.
Let $\mathit{vars}$ be some fixed integer values and $\inputv{s}_b$
be a free variable. Suppose that the query $b(\mathit{vars}, \inputv{s}_b)$
has an infinite computation in $P_{\mathit{CLP}}$.

First, suppose that $P$ does not contain any $\mathsf{call}$
instruction. Then, $P_{\mathit{CLP}}$ is constructed using
Definition~\ref{def:clp_block_no_call_not_void} only.
Let $P'_{\mathit{CLP}}$ be the $\clppl$ program constructed as
in~\cite{SpotoMP10}. The existence of an infinite computation
of $b(\mathit{vars},\inputv{s}_b)$ in $P_{\mathit{CLP}}$ entails
the existence of an infinite computation of
$b(\mathit{vars})$ in $P'_{\mathit{CLP}}$ which, by
Theorem~\ref{theorem:clp_completeness}, entails the existence
of a non-terminating execution of $J$ started at block $b$.

Now, suppose that $P$ contains a $\mathsf{call}$ instruction
to a method $m$. Then, the result follows from 
Propositions~\ref{proposition:composition},
\ref{proposition:instructions_exactness} and
\ref{proposition:instructions_exists_sigma}
% Theorem~\ref{theorem:clp_completeness}
and the fact that, in Definitions~\ref{def:clp_block_call_not_void}%
--\ref{def:clp_block_call_not_void_one_inst}, the
operational semantics of the call is modeled in 
$P_{\mathit{CLP}}$ by:
\begin{itemize}
\item the constraint $c_=$, which specifies that
  the path-length of the local variables and stack elements
  under the actual parameters is not modified by the call,
\item the constraint $\inputv{s}^{\#s-1}\geq 1$, which specifies
  that the receiver of the call is not $\mathtt{null}$,
\item the atom $b_m(\inputv{s}^{\#s-1+p-1},\ldots,\inputv{s}^{\#s-1},
  \outputv{s}^{\#s-1})$, where $b_m$ denotes the entry block of $m$
  and $\inputv{s}^{\#s-1+p-1},\ldots,\inputv{s}^{\#s-1}$ are
  the actual parameters of $m$ and 
  $\outputv{s}^{\#s-1}$ is the result of $m$,
\item if the block $\mathit{bb}$ where the call occurs consists of
  more than one instruction, clauses of the form
  \begin{align*}
    \mathit{bb}(\inputv{\vars},\inputv{s}_{\mathit{bb}}) \la \; &
    c_= \cup \{\inputv{s}^{\#s-1}\geq 1,\
    \inputv{s}_{\mathit{bb}}=\outputv{s}_{\mathit{bb}}\},\\
    & b_m(\inputv{s}^{\#s-1+p-1},\ldots,\inputv{s}^{\#s-1},
    \outputv{s}^{\#s-1}),\\
    & \mathit{bb}'(\outputv{\vars}',\outputv{s}_{\mathit{bb}})\\
    \mathit{bb}'(\inputv{\vars}',\inputv{s}_{\mathit{bb}}) \la \; &
    c\cup\{\inputv{s}_{\mathit{bb}}=\outputv{s}_{\mathit{bb}}\},
    b_1(\outputv{\vars},\outputv{s}_{\mathit{bb}})\\
    &\cdots\\
    \mathit{bb}'(\inputv{\vars}',\inputv{s}_{\mathit{bb}}) \la \; &
    c\cup\{\inputv{s}_{\mathit{bb}}=\outputv{s}_{\mathit{bb}}\}
    ,b_n(\outputv{\vars},\outputv{s}_{\mathit{bb}})
  \end{align*}
  where the call to $\mathit{bb}'$ in the first clause models the
  continuation of the execution after the call to $m$,
\item if the block $\mathit{bb}$ where the call occurs consists of
  exactly one instruction, clauses of the form
  \begin{align*}
    \mathit{bb}(\inputv{\vars},\inputv{s}_{\mathit{bb}}) \la \;
    & c_= \cup \{\inputv{s}^{\#s-1}\geq 1,\
    \inputv{s}_{\mathit{bb}}=\outputv{s}_{\mathit{bb}}\},\\
    & b_m(\inputv{s}^{\#s-1+p-1},\ldots,\inputv{s}^{\#s-1},
    \outputv{s}^{\#s-1}),\\
    & b_1(\outputv{\vars},\outputv{s}_{\mathit{bb}})\\
    \cdots & \\
    \mathit{bb}(\inputv{\vars},\inputv{s}_{\mathit{bb}}) \la \;
    & c_= \cup \{\inputv{s}^{\#s-1}\geq 1,\
    \inputv{s}_{\mathit{bb}}=\outputv{s}_{\mathit{bb}}\},\\
    & b_m(\inputv{s}^{\#s-1+p-1},\ldots,\inputv{s}^{\#s-1},
    \outputv{s}^{\#s-1}),\\
    & b_n(\outputv{\vars},\outputv{s}_{\mathit{bb}})
  \end{align*}
  where the calls to $b_1$, \ldots, $b_n$ model the continuation
  of the execution after the call to $m$.
\end{itemize}

\end{document}